\documentclass[11pt]{article}

\usepackage{amsmath}
\usepackage{amsthm}
\usepackage{amssymb}
\usepackage{graphicx}
\graphicspath{ {Figures/} }
\usepackage{color}
\usepackage[margin = 1in]{geometry} 
\usepackage{cite}
\usepackage[toc,page]{appendix}

\newcommand{\val}{\text{val}}
\newcommand{\ex}{\text{ex}}

\newtheorem{theorem}{Theorem}[section]
\newtheorem{lemma}{Lemma}[section]

\usepackage[mathlines, pagewise]{lineno}

\begin{document}
\title{Maximum Integer Flows in Directed Planar Graphs with Multiple Sources and Sinks and Vertex Capacities}

\author{Yipu Wang
			\\[1ex]
	University of Illinois at Urbana-Champaign\\
		{ywang298@illinois.edu}}

\maketitle

\begin{abstract}
	We consider the problem of finding maximum flows in planar graphs with capacities on both vertices and edges and with multiple sources and sinks. 
    We present three algorithms when the capacities are integers. 
    The first algorithm runs in $O(n \log^3 n + kn)$ time when all capacities are bounded, where $n$ is the number of vertices in the graph and $k$ is the number of terminals. This algorithm is the first to solve the vertex-disjoint paths problem in near-linear time when $k$ is bounded but larger than 2. The second algorithm runs in $O(k^2(k^3 + \Delta) n \text{ polylog} (nU))$ time, where $U$ is the largest finite capacity of a single vertex and $\Delta$ is the maximum degree of a vertex. 
    Finally, when $k=3$, we present an algorithm that runs in $O(n \log n)$ time; this algorithm works even when the capacities are arbitrary reals.
    Our algorithms improve on the fastest previously known algorithms when $k$ and $\Delta$ are small and $U$ is bounded by a polynomial in $n$.
    Prior to this result, the fastest algorithms ran in $O(n^2 / \log n)$ time for real capacities and $O(n^{3/2} \log n \log U)$ for integer capacities.
\end{abstract}

\section{Introduction}
Finding a maximum flow in a graph is a well-studied problem with applications in many fields.
The problem remains interesting even in planar graphs, which are graphs that can be embedded in the plane without crossing edges.
Such graphs arise in, for example, road traffic models and VLSI design.

Typically, the maximum flow problem asks us to route some commodity along edges with capacities, which limit the amount of commodity that can go through the edge.
In this paper we are concerned with the case where vertices of the graph also have capacities, which limit the amount of commodity that can go through that vertex.
When all the arc and vertex capacities are unit, we get the {\em vertex-disjoint paths problem}.

In general graphs, adding capacities to the vertices does not make the problem any harder because of a reduction first suggested by Ford and Fulkerson~\cite{FF62}.
For each vertex $v$ with finite capacity $c$, we do the following.
Replace $v$ with two vertices $v_{in}$ and $v_{out}$, and add an arc of capacity $c$ directed from $v_{in}$ to $v_{out}$.
All arcs that were directed into $v$ are directed into $v_{in}$ instead, and all arcs that were directed out of $v$ are directed out of $v_{out}$ instead.
Unfortunately, this reduction does not preserve planarity.
Consider $K_4$, the complete graph on four vertices.
If we apply the reduction of Ford and Fulkerson, we get a (directed) graph whose underlying undirected graph is $K_5$, which is not planar by Kuratowski's Theorem.

Prior work on this problem has focused on the case where there is a single source and sink or when the number of vertices with capacities is bounded.
Khuller and Naor~\cite{KN94} were the first to consider the case where there is a single source and sink. 
Currently, the best known algorithm for this case is due to Kaplan and Nussbaum ~\cite{KN11}, who described an algorithm for maximum flow in directed planar graphs with vertex capacities that runs in $O(n \log n)$ time.
In doing so, they fixed a flaw in a paper of Zhang, Liang and Chen~\cite{ZLC08}.
They also give an algorithm that runs in $O(n)$ time when all vertex and edge capacities are unit, solving the vertex-disjoint paths problem in directed planar graphs with a single source and sink.
Zhang, Liang, and Chen~\cite{ZLC08} described an algorithm that finds a maximum flow in undirected $st$-planar graph in $O(n)$ time. (A planar graph is $st$-planar if the source and the sink are on the same face.)

In the case of multiple sources and sinks, Borradaile et al. give an algorithm that runs in $O(\alpha^3 n \log^3 n)$ time, where $\alpha$ is the number of vertex capacities~\cite{BKMNW17}.
For arbitrary numbers of terminals and vertex capacities, the best-known algorithm prior to this paper uses the Ford-Fulkerson reduction described earlier, connects a super-source to all sources, connects all sinks to a supersink, and then in the resulting graph applies either Goldberg and Rao's algorithm~\cite{GR98} for finding maximum flows in networks with integer capacities or Orlin's algorithm~\cite{O13} for finding maximum flows in sparse graphs with real capacities.
For input graphs in which all vertex and arc capacities are integers, the resulting algorithm runs in $O(n^{3/2}\log n \log U)$ time where $U$ is the largest capacity; for input graphs with real capacities, the resulting algorithm runs in $O(n^2/\log n)$ time.

In this paper, we improve on these algorithms in some special cases by extending Kaplan and Nussbaum's algorithm to certain graphs with multiple sources and sinks.
First, we observe that when there are multiple sources and sinks, applying Kaplan and Nussbaum's algorithm results in a flow that is infeasible at only $k-2$ vertices.
For each of these infeasible vertices, we define the excess of the vertex to be the amount by which it is infeasible, and we show that the some of the excesses of all the infeasible vertices is at most $(k-2)U$.
This means that when $U$ is small, the flow returned by Kaplan and Nussbaum's algorithm is close to feasible.
We exploit this observation to obtain our first algorithm: when $U$ is bounded by a constant, the maximum flow can be found in $O(n \log n + kn)$ time.
When $k$ is bounded, this algorithm solves the vertex-disjoint paths problem in near-linear time.
More generally, we show that if the sum of the excesses of the infeasible vertices is $O(1)$, then we can get rid of the excesses in linear time.

Our second algorithm deals with the case where $U$ may be unbounded. The basic idea is a scaling algorithm.
First we guess the value of the maximum flow using binary search; this increases the running time of the algorithm by a factor $O(\log (nU))$.
Starting with a flow with $k-2$ infeasible vertices, we find a way to improve the flow that decreases the maximum excess of the vertices by some factor that depends only on $k$ and $\Delta$. The improved flow has the same value as the original flow. We show that after $O(k \log (kU))$ improvement phases, each infeasible vertex has excess at most $O(k\Delta)$.
Thus, when $k$ and $\Delta$ are small, we get a flow that is almost feasible.
As in the first algorithm, we exploit this observation to quickly eliminate the excesses to get the desired maximum flow.

Our third algorithm deals with the special case where $k=3$.
In this case, the fact that there is only one infeasible vertex considerably simplifies the problem, since we can just focus on decreasing the excess of this one vertex without worrying about trade-offs. (Roughly speaking, if there is more than one infeasible vertex, we have to consider that decreasing the excess of one vertex could increase the excess of another vertex.) We show that we can modify our second algorithm such that only one improvement phase is necessary.
This third algorithm works even if the capacities are arbitrary real numbers instead of integers.

The outline of this paper is as follows.
In section 2, we give some basic definitions and describe some basic graph constructions that will be used in the paper.
In section 3, we prove the structural properties that show that Kaplan and Nussbaum's algorithm almost works when there are multiple sources and sinks.
In section 4, we describe the algorithm for the case where capacities are bounded integers. 
In section 5, we use this algorithm to solve the case of arbitrary integer capacities.
In section 6, we describe the modifications to the algorithms that are necessary for the case when $k=3$ and the capacities are arbitrary reals.


\section{Preliminaries}
In this paper, $G$ is a simple directed plane graph with vertex set $V(G)$, arc set $E(G)$, and face set $F(G)$.
Let $n$ be the number of vertices in $G$; it is well known that Euler's formula implies $|E(G)| = O(n)$.
For any vertex $v \in V(G)$, let $\deg_G(v)$ denote the degree of $v$ in $G$, and let $\Delta$ be the largest degree in $G$.
If $G$ is a graph and $W \subseteq V(G)$, then $G \setminus W$ is the induced subgraph of $G$ with vertex set $V(G) \setminus W$.
For any integer $N$, let $[N] = \{1, \dots, N\}$.

We use $(u,v)$ to denote an arc or directed edge that is directed from $u$ to $v$.
A {\em path} is a sequence of arcs $((u_1, v_1), \dots, (u_p, v_p))$ such that $v_i = u_{i+1}$ for all $i \in [1, p-1]$.
Such a path {\em starts} at $u_1$ and {\em ends} at $v_p$.
If in addition $v_p = u_1$ then $P$ is a {\em cycle}.
A path $P$ contains a vertex $v$ if one of the edges of $P$ has $v$ as an endpoint.
Thus we will sometimes view paths and cycles as sets of vertices or as sets of arcs instead of as sequences of arcs.
For any $v \in V$, let $in(v) = \{(u,v) \mid (u,v) \in E(G)\}$ be the set of {\em incoming arcs} of $v$, and let $out(v) = \{(v,u) \mid (v,u) \in E(G)\}$ be the set of {\em outgoing arcs} of $v$.
Similarly, if $W$ is a set of vertices, then $in(W) = \{(u,v) \in E(G) \mid u \notin W, v \in W\}$ and $out(W) = \{(u,v) \in E(G) \mid u \in W, v \notin W\}$.

The {\em reversal} of any edge $(u,v)$, denoted $rev((u,v))$, is $(v,u)$.
We may assume without loss of generality that if $e \in E(G)$, then $rev(e) \in E(G)$.
If $P$ is a path $(e_1, \dots, e_p)$, then the {\em reversal} of $P$, denoted $rev(P)$, is $(rev(e_p), \dots, rev(e_1))$.

Two disjoint subsets of $V(G)$ are special: $S$ is a set of {\em sources} and $T$ is a set of {\em sinks} or {\em targets}.
Vertices that are in either $S$ or $T$ are called {\em terminals}.
Let $k$ be the number of terminals.
We may assume without loss of generality that none of the sources have incoming edges and none of the sinks have outgoing edges.

Each arc $e$ has a positive capacity $c(e)$ and each non-terminal vertex $v$ has a positive capacity $c(v)$.
Capacities may be infinite, and we can assume without loss of generality that terminals have infinite capacity: if a source $s$ has finite capacity $c$, then we can add a node $s'$, an edge $(s', s)$ of capacity $c$, replace $s$ with $s'$ in $S$, and let $s'$ have infinite capacity, all while preserving planarity.
A similar reduction eliminates finite capacities on the sinks.

{\bf Flows.}
A {\em flow network} is a directed graph that has a capacity on each arc and vertex, a set of sources, and a set of sinks.
Suppose $G$ is a flow network with capacity function $c : E(G) \cup V(G) \to [0, \infty)$, source set $S$, and target set $T$. 
Let $f : E(G) \to [0, \infty)$. 
To lighten notation, in this paper we will write $f(u,v)$ instead of $f((u,v))$ for any arc $(u,v)$. 
For each vertex $v$, let 
\[
	f^{in}(v) = \sum_{e \in in(v)} f(e) \;\;\; \text{and} \;\;\; f^{out}(v) = \sum_{e \in out(v)} f(e).
\]
Similarly, if $W$ is a set of vertices, then let
\[
	f^{in}(W) = \sum_{e \in in(W)} f(e) \;\;\; \text{and} \;\;\; f^{out}(W) = \sum_{e \in out(W)} f(e).
\]
The function $f$ is a {\em flow in $G$} if it satisfies the following {\em flow conservation constraints}:
\[
	f^{in}(v) = f^{out}(v)  \;\;\; \forall v \in V(G) \setminus  (S \cup T)
\]

A flow is {\em feasible} if in addition it satisfies the following two types of constraints:
\[
	\begin{array}{ll}
		0 \leq f(e) \leq c(e) & \forall e \in E(G)\\
		f^{in}(v) \leq c(v) & \forall v \in V(G) \setminus (S \cup T)\\
	\end{array}
\]
Constraints of the first type are {\em arc capacity constraints} and those of the second type are {\em vertex capacity constraints}.
A flow $f$ {\em routes} $f(e)$ units of flow through the arc $e$. 
An arc $e \in in(v)$ {\em carries flow into} $v$ if $f(e) > 0$, and an arc $e' \in out(v)$ {\em carries flow out of} $v$ if $f(e') > 0$.
We assume that $\min\{f(e), f(rev(e)\} = 0$ for every edge $e$.

In the {\em maximum flow problem}, we are trying to find a feasible flow $f$ with maximum {\em value}, where the value $v(f)$ of a flow $f$ is defined as
\[
	v(f) = \sum_{s \in S} f^{out}(s).
\]
When all the vertex and arc capacities are 1, the maximum flow problem becomes the {\em vertex-disjojint paths problem}.

Let $\val(G)$ be the value of the maximum flow in a flow network $G$ (which may have vertex capacities).
A {\em circulation} is a flow of value 0.
A circulation $g$ is {\em simple} if $g^{in}(v) = g^{out}(v)$ for every terminal $v$.
Non-simple circulations only exist if there are more than two terminals.
A flow $f$ has a {\em flow cycle} $C$ if $C$ is a cycle and $f(e) > 0$ for every arc $e$ in $C$, and $f$ is {\em acyclic} if it has no flow cycles.
A flow cycle $C$ of a flow $f$ is {\em unit} if $f(e) = 1$ for every arc $e$ in $C$.
A flow $f$ {\em saturates} an arc $e$ if $f(e) = c(e)$.
A flow is a {\em path-flow} if its support is a path.

We will often add two flows $f$ and $g$ together to obtain a flow $f + g$, or multiply a flow $f$ by some constant $c$ to get a flow $cf$.
These operations are defined in the obvious way: for every arc $e$, we have
\begin{align*}
	(f+g)(e) &= \max\{0, f(e) + g(e) - f(rev(e)) - g(rev(e))\}\\
	(cf)(e) &= c \cdot f(e)
\end{align*}

{\bf Apex graphs and $G_{st}$.}
A graph $G$ is a {\em $k$-apex graph} if there are at most $k$ vertices whose removal from the graph would make $G$ planar.
These $k$ vertices are called {\em apices}. 

Given a flow network with multiple sources and sinks, we can reduce the maximum flow problem to the single-source, single-sink case by adding a supersource $s$, supersink $t$, infinite-capacity arcs $(s, s_i)$ for every $s_i \in S$, and infinite-capacity arcs $(t_i, t)$ for every $t_i \in T$.
Call the resulting flow network $G_{st}$.
Finding a maximum flow in the original network $G$ is equivalent to finding a maximum flow from $s$ to $t$ in $G_{st}$.
The graph $G_{st}$ is not necessarily planar but is a 2-apex graph.

{\bf The flow graph $f_G$.}
Given a flow $f$ in a flow network $G$, the {\em flow graph} of $f$ is a graph $f_G$ with the same vertex and arc set as $G$, but each arc $e$ in $f_G$ has weight $f(e)$.
Depending on the context, we will interpret these arc weights as either capacities or flow.

{\bf The extended graph $G^\circ$.} Given a flow network $G$ with vertex capacities, Kaplan and Nussbaum \cite{KN11} defined the {\em extended graph} $G^\circ$ based on constructions of Khuller and Naor~\cite{KN94}, Zhang, Liang, and Jiang~\cite{ZLJ06}, and Zhang, Liang, and Chen~\cite{ZLC08}. 
Starting with $G_{st}$, we replace each finitely capacitated vertex $v \in V(G_{st})$ with an undirected cycle of $d$ vertices $v_1, \dots, v_d$, where $d$ is the degree of $v$. 
Each edge in the cycle has capacity $c(v)/2$. (An undirected edge $e$ with capacity $c(e)$ can be viewed as two arcs $e$ and $rev(e)$, each with capacity $c(e)$, so $G^\circ$ can be viewed as a directed flow network.)
We make every edge that was incident to $v$ incident to some vertex $v_i$ instead, such that each edge is connected to a different vertex $v_i$, the clockwise order of the edges is preserved, and the graph remains planar. 
We also identify the new arc $(u,v_i)$ or $(v_i, u)$ with the old arc $(u,v)$ or $(v,u)$ and denote the cycle replacing $v$ by $C_v$. The graph $G^\circ$ has $O(n)$ vertices and arcs.
See Figure~\ref{F:auxiliary-graphs}.

This idea of eliminating vertex capacities in planar graphs by replacing each vertex with a cycle has also been used in the context of finding shortest vertex-disjoint paths in planar graphs~\cite{CdVS11}.	
    
    {\bf The graph $\overline{G}$.} 
    Given a flow network $G$ with vertex capacities, let $\overline{G}$ be the flow network obtained as follows: Starting with $G_{st}$, replace each capacitated vertex $v$ with two vertices $v^{in}$ and $v^{out}$, and add an arc of capacity $c(v)$ directed from $v^{in}$ to $v^{out}$.
All arcs that were directed into $v$ are directed into $v^{in}$ instead, and all arcs that were directed out of $v$ are directed out of $v^{out}$ instead.
See Figure~\ref{F:auxiliary-graphs}. It is well known that every feasible flow in $G_{st}$ f corresponds to a feasible flow in $\overline{G}$ of the same value, and vice versa.
The graph $\overline{G}$ has $O(n)$ vertices and arcs.
\begin{figure}
\centering
\begin{tabular}{cr@{\qquad}cr@{\qquad}cr}
	\includegraphics[scale=0.5]{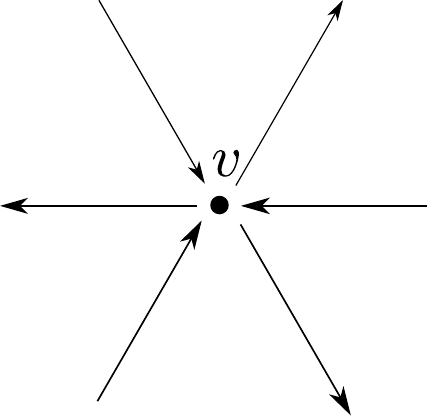} & \hspace{-0.25in}(a)
	&
	\includegraphics[scale=0.5]{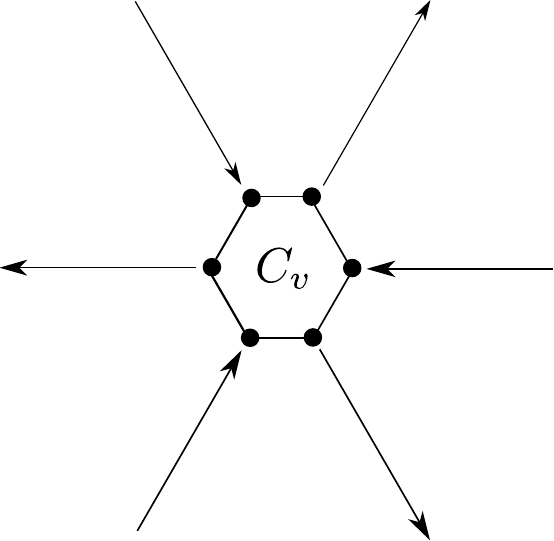} & \hspace{-0.25in}(b)
    &
	\includegraphics[scale=0.5]{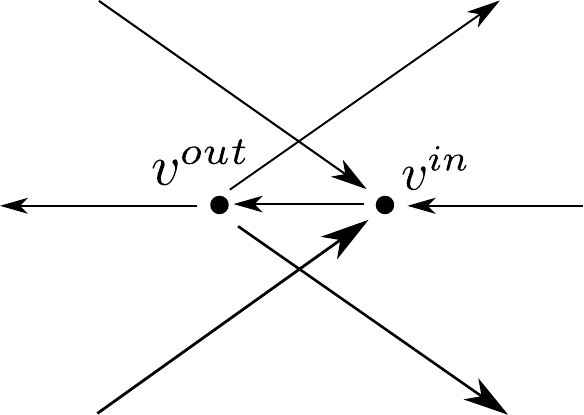} & \hspace{-0.25in}(c)
\end{tabular}
\caption{(a) capacitated vertex $v \in G$ with capacity $c(v)$ (b) corresponding cycle $C_v$ in $G^\circ$; each arc in $C_v$ has capacity $c(v)/2$ (c) corresponding arc $(v^{in}, v^{out})$ in $\overline{G}$ with capacity $c(v)$}
\label{F:auxiliary-graphs}
\end{figure}

{\bf Restrictions and extensions.}  
Suppose $G$ and $H$ are flow networks such that every arc in $G$ is also an arc in $H$. 
If $f'$ is a flow in $H$, then the {\em restriction} of $f'$ to $G$ is the flow $f$ in $G$ defined by $f(e) = f'(e)$ for all arcs $e \in E(G)$.
Conversely, if $f$ is a flow in $G$, then an {\em extension} of $f$ is any flow $f'$ in $H$ such that $f(e) = f'(e)$ for every edge $e \in E(G)$.

Every arc in $G$ or $G_{st}$ is an arc in both $\overline{G}$ and $G^\circ$.
Every feasible flow in $\overline{G}$ has a feasible restriction in $G$.
Conversely, every feasible flow $f$ in $G$ has a feasible extension $\overline{f}$ in $\overline{G}$, by defining $\overline{f}(v^{in}, v^{out}) = f^{in}(v)$.
Every feasible flow in $G^\circ$ has a restriction in $G$; this restriction is a flow but is not necessarily feasible.  
On the other hand, we have the following lemma:
\begin{lemma}\label{L:extendable}
	 Every feasible flow $f$ in $G$ has an extension $f^\circ$ that is feasible in $G^\circ$.
\end{lemma}
\begin{proof}
	We use the well-known flow decomposition theorem, which states that any flow $f$ in $G$ can be decomposed into a sum of flows $f_1, \dots, f_m$ such that for each $i$, the support of $f_i$ is either a cycle or a path from a source to a sink. 
    For each $i \in [m]$, let $p_i$ be the support of $f_i$ and let $u_i = v(f_i)$.
    
    For each capacitated vertex $w \in G$, we define $f^\circ$ on the cycle $C_w$ in $G^\circ$ as follows: for each $i \in [m]$, if some edge in $p_i$ carries $u_i$ units of flow into a vertex $x$ on $C_w$ and another edge in $p_i$ carries $u_i$ units of flow out of a vertex $x'$ on $C_w$, then we route $u_i/2$ units of flow clockwise along $C_w$ from $x$ to $x'$ and $u_i/2$ units of flow counter-clockwise along $C_w$ from $x$ to $x'$.
    It is easy to see that $f^\circ$ satisfies conservation constraints.
    Since $f^{in}(C_w) \leq c(w)$, no arc on $C_w$ carries more than $c(w)/2$ units of flow, so $f^\circ$ is feasible.
\end{proof}

We now describe how to convert a feasible flow $f$ in $G$ to a feasible extension $f^\circ$ of $f$ to $G^\circ$.
We must define $f^\circ(e) = f(e)$ for all arcs $e \in E(G)$.
We reduce the problem of finding $f^\circ$ on all other arcs to finding a flow in a flow network $H$.
Let $H$ be the subgraph of $G^\circ$ consisting of all cycles $C_v$ where $v$ is a capacitated vertex in $G$; it suffices to define $f^\circ$ on the arcs of $H$.
Recall that for all $v \in V(G)$, the vertices in $C_v$ are $v_1, \dots, v_d$ in clockwise order, where $d = \deg_G(v)$.
For each vertex $v_i$ in $H$, let $e_{i,v}$ be the unique arc in $G$ incident to $v_i$.
When it is clear what vertex $v$ is, we will write $e_i$ instead of $e_{i,v}$.
For each $v \in V(G)$ and $i \in [\deg_G(v)]$, let 
\[ 
	demand(v_i) = 
	\begin{cases} 
      -f(e_i) & \text{ if } e_i \in in(v_i) \\
      f(e_i) & \text{ if } e_i \in out(v_i) 
   \end{cases}
\]
That is, $demand(x_i)$ is the net amount of flow that $f^\circ$ carries out of $v_i$ so far.
For each vertex $v_i$ such that $demand(v_i)$ is negative, let $v_i$ be a source in $H$; similarly, if $demand(v_i)$ is positive, let $v_i$ be a sink in $H$.
For each $v \in V(G)$, $\sum_{i=1}^{\deg_G(v)} demand(v_i) = 0$.
See Figure~\ref{F:demands}.

\begin{figure}
\centering
\begin{tabular}{cr@{\qquad}cr@{\qquad}cr}
	\includegraphics[scale=0.75]{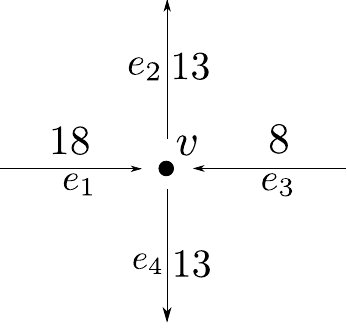} & \hspace{-0.25in} (a)
	&
	\includegraphics[scale=0.5]{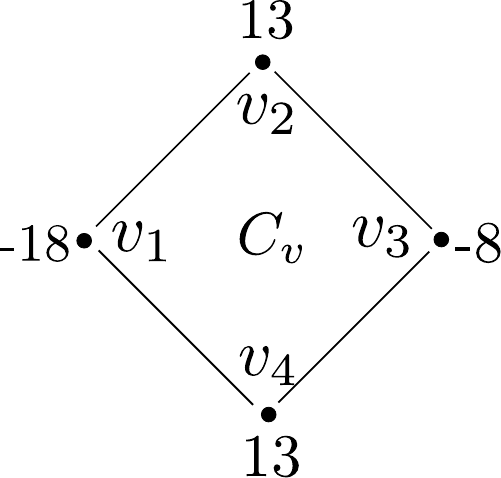} & (b)
\end{tabular}
\caption{Extending a flow from $G$ to $G^\circ$. (a) An example of $f$ at $v$; arcs are labeled with their flow values (b) $H$ at $C_v$ with terminals labeled with their demand values; $v_1$ and $v_3$ are sources; $v_2$ and $v_4$ are sinks}
\label{F:demands}
\end{figure}

By Lemma~\ref{L:extendable}, there exists a flow $f_H$ in $H$ such that $f_H^{out}(v_i) = -demand(v_i)$ for every source $v_i$ and $f_H^{in}(v_i) = demand(v_i)$ for every sink $v_i$.
To actually find $f_H$, we do the following. For each source $v_i$ in $H$, we add a vertex $v_i'$ that will be a source instead of $v_i$, and we add an arc $(v_i', v_i)$ with capacity $-demand(v_i)$; similarly, for each sink $v_j$ in $H$, we add a vertex $v_j'$ that will be a sink instead of $v_j$, and we add an arc $(v_j, v_j')$ with capacity $demand(v_j)$.
Then $f_H$ is an acyclic maximum flow in the resulting network.
The restriction of $f_H$ to $H$ is exactly $f^\circ$ on the arcs of $H$.
Finding $f_H$ requires finding a maximum flow in a planar graph with multiple sources and sinks, which can be done in $O(n \log^3 n)$ time using the algorithm of Borradaile et al.~\cite{BKMNW17}.

{\bf The residual graph.}
If $f$ is a flow in a flow network $G$ with capacity function $c$ and without vertex capacities, then the {\em residual capacity} of an arc $e$ with respect to $f$ and $c$, denoted $c_f(e)$, is $c(e) - f(e) + f(rev(e))$.
The {\em residual graph of $G$ with respect to $f$ and $c$} (or just the {\em residual graph of $G$ with respect to $f$} when $c$ is the capacity function given as input) has the same vertices and arcs as $G$, but each arc $e$ has capacity $c_f(e)$.
A {\em residual edge} of $G$ with respect to $f$ is an edge with positive residual capacity, a {\em residual path} is a path made up of residual edges, and a {\em residual cycle} is a cycle made up of residual edges.
It is well known that a flow $f$ is a maximum flow in a graph $G$ if the residual graph of $G$ with respect to $f$ does not have any residual paths from a source to a sink.

{\bf Fractional and integer flows.}
A flow $f^\circ$ in $G^\circ$ is an {\em integer flow} if $f^\circ(e)$ is an integer for every arc $e$ in $G^\circ$; otherwise, $f^\circ$ is {\em fractional}. 
The following lemma is well-known; for a proof, see Appendix~\ref{A:fractional}.
\begin{lemma}\label{L:fractional}
	Let $f^\circ$ be a fractional flow in a flow network $G^\circ$ such that $v(f^\circ)$ is an integer, $G^\circ$ has integer arc capacities, and $G^\circ$ has no vertex capacities. Then there exists an integer flow $f_1^\circ$ in $G^\circ$ of the same value as $f$ such that $|f^\circ(e) - f_1^\circ(e)| < 1$ for every arc $e$ in $G^\circ$.
\end{lemma}

We now describe how to convert a fractional flow $f^\circ$ in $G^\circ$ to an integer flow $f_1^\circ$ in $G^\circ$ of the same value, assuming that $v(f^\circ)$ is an integer, $G^\circ$ has integer arc capacities, and $G^\circ$ has no vertex capacities.
By Lemma~\ref{L:fractional}, we can assume that $|f^\circ(e) - f_1^\circ(e)| < 1$ for every arc $e$ in $G^\circ$.
Thus we initially define $f_1^\circ(e) = \lfloor f^\circ(e) \rfloor$ for all arcs $e$ in $G^\circ$; now we just need to increase $f_1^\circ(e)$ by 1 for some arcs $e$ in order to make $f_1^\circ$ satisfy conservation constraints and to make $v(f_1^\circ) = v(f^\circ)$.

We reduce the problem of fixing $f_1^\circ$ to finding a flow in a flow network $H$.
Let $H$ be the subgraph of $G^\circ$ consisting of all arcs $e$ where $f^\circ(e)$ has a non-zero fractional part.
All arcs in $H$ have capacity 1.
For each vertex $v$ in $G^\circ$, let 
\[ 
	demand(v) = 
    \begin{cases} 
      (f_1^\circ)^{out}(v) - (f_1^\circ)^{in}(v) & \text{ if } v \notin \{s,t\} \\
      (f_1^\circ)^{out}(v) - (f_1^\circ)^{in}(v) - v(f) & \text{ if } v = s \\
      (f_1^\circ)^{out}(v) - (f_1^\circ)^{in}(v) + v(f) & \text{ if } v = t 
   \end{cases}
\]
That is, $demand(v)$ is the net amount of flow that $f_1^\circ$ carries out of $v$ so far, minus the net amount of flow that $f_1^\circ$ is supposed to carry out of $v$.
For each vertex $v$ such that $demand(v)$ is negative, let $v$ be a source in $H$; similarly, if $demand(v)$ is positive, let $v$ be a sink in $H$.
We have $\sum_{v \in V(G^\circ)} demand(v) = 0$.

By Lemma~\ref{L:fractional}, there exists a flow $f_H$ in $H$ such that $f_H^{out}(v) = -demand(v)$ for every source $v$ and $f_H^{in}(v) = demand(v)$ for every sink $v$.
To actually find $f_H$, we do the following. For each source $v$ in $H$, we add a vertex $v'$ that will be a source instead of $v$, and we add an arc $(v', v)$ with capacity $-demand(v)$; similarly, for each sink $v$ in $H$, we add a vertex $v'$ that will be a sink instead of $v$, and we add an arc $(v, v')$ with capacity $demand(v)$.
We set $f_H$ to be the restriction to $H$ of any acyclic maximum flow in the resulting network.
To fix $f_1^\circ$, we just need to replace it with $f_1^\circ + f_H$. 
Finding $f_H$ requires finding a maximum flow in a 2-apex graph with multiple sources and sinks, which can be done in $O(n \log^3 n)$ time using the algorithm of Borradaile et al.~\cite{BKMNW17}.
We have proved the following lemma.
\begin{lemma}\label{L:fractional-alg}
	Let $f^\circ$ be a fractional flow in a flow network $G^\circ$ such that $v(f)$ is an integer, $G^\circ$ has integer arc capacities, and $G^\circ$ has no vertex capacities. Then in $O(n \log^3 n)$ time we can find an integer flow $f_1^\circ$ of the same value as $f$ such that $|f^\circ(e) - f_1^\circ(e)| < 1$ for every arc $e$ in $G^\circ$.
\end{lemma}

{\bf Subroutines.}
Our algorithm uses several algorithms that compute maximum flows or circulations in graphs without vertex capacities.
First, we use an algorithm of Borradaile et al.~\cite{BKMNW17} for finding maximum flows in directed planar graphs with multiple sources and targets in $O(n \log^3 n)$ time.
Equivalently, the algorithm finds maximum flows in directed planar graphs with a single source and sink if the source and sink are the only apices.
Second, we use another algorithm by Borradaile et al.~\cite{BKMNW17} that finds maximum flows in $k$-apex graphs with multiple sources and sinks in $O(k^3 n \log^3 n)$ time.
Third, we use the classical Ford-Fulkerson augmenting-path algorithm that computes maximum flows in general graphs with integer capacities in $O(mU^*)$ time, where $m$ is the number of edges in the flow network and $U^*$ is the value of the maximum flow.
Finally, we implicitly use two algorithms that allow us to assume without loss of generality that certain flows are acyclic.
The first is by Kaplan and Nussbaum~\cite{KN11}:
\begin{lemma}\label{L:cancel-cycles}
Given a feasible flow $f^\circ$ in $G^\circ$, we can compute in $O(n)$ time another feasible flow of the same value as $f^\circ$ whose restriction to $G$ is feasible and acyclic by canceling flow-cycles.
\end{lemma}
We describe this algorithm in more detail in Appendix~\ref{A:cancel-cycles}.
Using this algorithm, we can assume that whenever we compute a flow in $G^\circ$, the restriction of that flow to $G$ is acyclic.

The second algorithm that we use implicitly is by Sleator and Tarjan~\cite{ST83}:
\begin{lemma}
Given a flow in a flow network with $O(n)$ vertices and arcs, we can compute another flow of the same value that is acyclic in $O(n \log n)$ time by canceling flow-cycles.
\end{lemma}
Using this algorithm, we may assume that whenever we compute a flow in a graph, the computed flow is acyclic.


\section{Saddles and excess}
Suppose $f^\circ$ is a feasible flow in $G^\circ$ whose restriction $f$ to $G$ is acyclic. 
It is easy to see that $f$ satisfies conservation and arc capacity constraints.
In this section, we show that $f$ violates at most $k_1 + k_2 - 2$ vertex capacity constraints.

Let $f_G$ be the flow graph of $f$.
For any vertex $v$ in $f_G$, the {\em alternation number} of $v$, denoted by $\alpha(v)$, is the number of direction changes (i.e., from in to out or vice versa) of the arcs incident to $v$ as we examine them in clockwise order.
Thus $\alpha(u) = 0$ for all terminals $u$, and the alternation number of any vertex is even. 
A vertex $v$ is a {\em saddle in $f$} if $\alpha(v) \geq 4$. 
We let $index(v)$ denote the {\em index} of $v$ and define it by $index(v) = \alpha(v)/2 - 1$.

Guattery and Miller~\cite{GM92} showed the following:

\begin{lemma}\label{L:few-saddles}
	If $f_G$ is a plane directed acyclic graph with $k_1$ sources and $k_2$ sinks, then the sum of the indices of the saddles in $f_G$ is at most $k_1+k_2-2$. 
\end{lemma}
In particular, a vertex in $f_G$ is a saddle if and only if it has positive index, so $f_G$ has at most $k - 2$ saddles.
A proof of Lemma~\ref{L:few-saddles} can be found in Appendix~\ref{A:saddles-and-excess}.

A vertex $v \in V(G)$ is {\em infeasible} under a flow $f$ if $f^{in}(v) > c(v)$ and {\em feasible} otherwise.
For any vertex $v \in V(G)$, let $\ex(f^\circ, v)$ and $\ex(f,v)$ denote the {\em excess} of the vertex $v$ under $f^\circ$ or $f$:
\[
	\ex(f^\circ, v) = \ex(f,v) = \max\{0, f^{in}(v) - c(v)\}
\]
The excess of a vertex is positive if and only if the vertex is infeasible. We also define $\ex(f^\circ) = \ex(f) = \max_{v \in V(G)} \ex(f,v)$.
We will sometimes say that $f$ has excess $\ex(f,v)$ on $v$.

\begin{lemma}\label{L:small-excess}
	Let $index(v)$ be defined for each vertex $v$ in $G$ using the flow graph $f_G$ of $f$.
    For each vertex $v$ in $f_G$, we have $\ex(f,v) \leq index(v) c(v)$.
\end{lemma}
\begin{proof}
	Let $f^\circ_G$ be the flow graph of $f^\circ$.
	We have $\alpha(v) = 2 \cdot index(v) + 2$.
    Thus, if we examine the arcs in $f_G$ incident to $v$ in clockwise order, there are $index(v)+1$ groups of consecutive incoming arcs.
    Consider such a group of consecutive incoming arcs $(u_i,v), \dots, (u_j, v)$ in $f_G$. We can view these as arcs $(u_i, v_i), \dots, (u_j, v_j)$ in $f^\circ_G$, where $v_i, \dots, v_j$ are consecutive vertices in $C_v$. 
    In $f^\circ_G$, the only two arcs in $out(\{v_i, \dots, v_j\})$ are $(v_i, v_{i-1})$ and $(v_j, v_{j+1})$, which have total capacity $c(v)$.
    Thus, for each vertex $v$ in $f_G$, each group of consecutive incoming arcs in $f_G$ has total weight at most $c(v)$.
    This shows that $f^{in}(v) \leq (index(v)+1)c(v)$ for any vertex $v$, from which the lemma follows.
\end{proof}
Combining Lemmas~\ref{L:few-saddles} and~\ref{L:small-excess}, we see that the sum of the excesses of all vertices under $f$ is $(k-2)U$.
Lemma~\ref{L:small-excess} implies that $f$ is only infeasible at saddles of $f_G$.

\section{Bounded integer capacity case}\label{S:bounded-capacity}
Suppose that all vertex and arc capacities are integers less than some constant $U$.
Let $f^\circ$ be an integral maximum flow in $G^\circ$, and let $f$ be its restriction to $G$. By Lemma~\ref{L:cancel-cycles} we may assume without loss of generality that $f$ is acyclic.
The flow $f$ may be infeasible at up to $k-2$ vertices $x_1, \dots, x_{k-2}$.
By Lemma~\ref{L:few-saddles} and~\ref{L:small-excess}, the sum of the excesses of the infeasible vertices is at most $(k-2)U$.
Computing $f$ takes $O(n \log^3 n)$ time using the algorithm of Borradaile et al.~\cite{BKMNW17}.
After finding $f$, the algorithm has two steps.

{\bf Step 1.} In this step, we remove $\ex(f,x)$ units of flow through each infeasible vertex $x$ to get a feasible flow $f_1$ in $G$.
The flow $f_1$ will have lower value than $f$.
To do this, let $f_G$ be the flow graph of $f$.
The graph $f_G$ is a directed acyclic graph.
To remove one unit of flow through an infeasible vertex $x$, we do the following:
\begin{itemize}
	\item Find a path $P_s$ in $f_G$ from $s$ to $x$, and a path $P_t$ in $f_G$ from $x$ to $t$.
Since $f_G$ is acyclic, $P_s$ and $P_t$ are internally disjoint.
	\item The arcs in $P_s \cup P_t$ form a path from $s$ to $t$. For every arc $e$ of $P_s \cup P_t$, decrease $f(e)$ by 1. The resulting $f$ is a flow in $G$ whose value has been decreased by one and whose excess through $x$ has been decreased by one. We also update $f_G$ accordingly.
\end{itemize}
By Lemmas~\ref{L:few-saddles} and~\ref{L:small-excess}, we only need to remove $(k-2)U$ units of flow from $f$ in order for $f$ to become a feasible flow in $G$.
Let $f_1$ be the resulting feasible flow.
Finding $P_s$ and $P_t$ and updating $f$ on all edges of $P_s \cup P_t$ takes $O(n)$ time, so step 1 runs in $O(knU)$ time.

{\bf Step 2.} 
Let $\overline{f_1}$ be the extension of $f_1$ to $\overline{G}$. 
In this step, we do the following:
\begin{itemize}
	\item Compute a maximum flow $\overline{f_2}$ in the residual graph of $\overline{G}$ with respect to $\overline{f_1}$ using the classical Ford-Fulkerson algorithm. 
	\item Return the restriction of $\overline{f_1} + \overline{f_2}$ to $G$.
\end{itemize}
Since $\overline{f_2}$ is a maximum flow in the residual graph of $\overline{G}$ with respect to $\overline{f_1}$, we see that $\overline{f_1} + \overline{f_2}$ is a maximum flow in $\overline{G}$.
It follows that the restriction of $\overline{f_1} + \overline{f_2}$ to $G$ is a maximum flow in $G$, as desired.

We have $val(\overline{G}) \leq val(G^\circ) = v(f) \leq v(f_1) + (k-2)U$.
Thus the value of $\overline{f_2}$ is at most $(k-2)U$, so computing $\overline{f_2}$ takes $O(knU)$ time. Hence step 2 takes $O(knU)$ time.
Thus, if $U$ is a constant, the entire algorithm runs in $O(n \log^3 n + kn)$ time.

\section{Integer capacities and $k > 2$}\label{S:integer}
Suppose all vertex and arc capacities are integers.
Let $\lambda^* = val(G)$.
The basic structure of the algorithm is as follows:
\begin{itemize}
	\item Guess $\lambda^*$ via binary search. 
    \begin{enumerate}
    	\item Suppose we guess the value of the maximum flow of $G$ to be $\lambda$. Find a maximum flow $f^\circ$ in $G^\circ$ of value $\lambda$. By Lemma~\ref{L:cancel-cycles}, we may assume that the restriction $f$ of $f^\circ$ to $G$ is acyclic.
    	\item While $\ex(f) > 2k \Delta$, improve $f$.
    	\item Fix $f$ using the algorithm from section~\ref{S:bounded-capacity}.
    \end{enumerate}
\end{itemize}
One can see that the algorithm has three main steps which we call {\em phases}.
In phase 2, improving $f$ means that we find a flow $f_1$ of the same value as $f$ such that 
\[
	\ex(f_1) \leq \frac{k-1}{k} ex(f) + \Delta.
\]
We then set $f$ to be the new flow $f_1$.
We will eventually show that a single improvement of $f$ can be done in $O(k^4 n \log^3 n)$ time.
In phase 3, fixing $f$ means that we remove $\ex(f,x)$ units of flow through each infeasible vertex $x$ to get flow $f'$, extend $f'$ to a flow $\overline{f'}$ in $\overline{G}$, and then use the Ford-Fulkerson algorithm to find a maximum flow $\overline{f''}$ in the residual graph of $\overline{G}$ with respect to $\overline{f'}$; we then set $f$ to be the restriction of $\overline{f'} + \overline{f''}$ to $G$.
Regarding the binary search, $\lambda \leq \lambda^*$ if the result of phase 3 is a feasible flow of value
$\lambda$, and $\lambda > \lambda^*$ if either phase 2 fails or if the flow that results from phase 3
has value less than $\lambda$.

Before we describe how phase 2 is implemented, let us analyze the running time of the algorithm.
If $U$ be the maximum capacity of a single vertex, then $\lambda^* \leq nU$, so the binary search for $\lambda^*$ only requires $O(\log (nU))$ guesses.
Computing $f^\circ$ in phase 1 takes $O(n \log^3 n)$ time using the algorithm of Borradile et al.~\cite{BKMNW17}.
By Lemma~\ref{L:small-excess}, at the beginning of phase 2, $\ex(f) \leq (k-2)U$.
The following lemma shows that phase 2 takes $O(k^5 n \log^3 n \log (kU))$ time:
\begin{lemma}
	After $O(k \log (kU))$ iterations of the while-loop in phase 2, $\ex(f,x) \leq 2k\Delta$ for every vertex $x \in V(G)$.
\end{lemma}
\begin{proof}
	After each iteration, $\ex(f)$ decreases roughly by a factor $1 + 1/(k-1) \geq 1 + 1/k$.
    Thus we only require $O(\log_{1 + 1/k} (kU)) = O(\frac{\ln (kU)}{\ln(1 + 1/k)})$ iterations.
    For $k \geq 1$ we have
    \begin{align*}
    	e^{1/2} < &(1 + 1/k)^k < e\\
        \implies 1/2 < &k \ln (1 + 1/k) < 1\\
        \implies \frac{1}{2k} < &\ln (1 + 1/k) < \frac{1}{k}.
    \end{align*}
    This means that $O(k \log kU)$ iterations suffice.
\end{proof}
In phase 3, the same reasoning as in Section~\ref{S:bounded-capacity} shows that computing $\overline{f'} + \overline{f''}$ takes $O(k^2n\Delta)$ time.
The total running time of the algorithm is thus
\begin{align*}
	O(\log (nU) [n \log^3 n + k^5 n \log^3 n \log (kU) + k^2n\Delta]) &= O(k^2 n (k^3 \log^3 n \log kC + \Delta) \log nC)\\
    &= O(k^2(k^3 + \Delta)n \text{ polylog}(nU)).
\end{align*}

The rest of this section describes one iteration of the while-loop in phase 2.
Specifically, given a feasible flow $f^\circ$ whose restriction $f$ to $G$ has at most $k-2$ infeasible vertices, we compute a flow $f_1^\circ$ in $G^\circ$ whose restriction $f_1$ to $G$ has at most $k-2$ infeasible vertices, each of which has excess at most $\frac{k-1}{k}ex(f) + \Delta$.
Let $X$ be the set of infeasible vertices under $f$, and for each $x \in X$, define $ex_x = \ex(f, x)$.
The procedure that finds $f^\circ_1$ has two stages, and in each stage we are trying to find a circulation in $G^\circ$ that can be added to $f^\circ$ to get $f_1^\circ$.
In stage 1, we find a circulation $g^\circ$ such that $f^\circ + g^\circ$ is feasible in $G^\circ$ and $(f^\circ + g^\circ)^{in}(C_x) \leq c(x)$ for every $x \in X$. However, the restriction of $f^\circ + g^\circ$ to $G$ may have large excesses on vertices not in $X$. To fix this, in stage 2 we use $g^\circ$ to compute a circulation $g_k^\circ$ such that $f^\circ + g_k^\circ$ is a feasible flow in $G^\circ$ and $\ex(f^\circ + g_k^\circ) \leq \frac{k-1}{k} \ex(f) + \Delta$. Intuitively, $g_k^\circ$ approximates $g^\circ/k$ while being an integer circulation. In stage 3, we use Lemma~\ref{L:cancel-cycles} so that without loss of generality we can assume the restriction of $f^\circ + g_k^\circ$ to $G$ is acyclic and has at most $k-2$ infeasible vertices; we then set $f_1^\circ = f^\circ + g_k^\circ$. If $\lambda > \lambda^*$, then $g^\circ$ may not exist and stage 1 may fail; if $\lambda \leq \lambda^*$, then $g^\circ$ exists and all three stages will work. 

\subsection{Stage 1}\label{S:integer-h}
To get $g^\circ$, we first convert $f^\circ$ to a feasible flow $f^\times$ of the same value in a flow network $G^\times$ such that the restrictions of $f^\circ$ and $f^\times$ to $G$ are equal.
Then, we find a circulation $g^\times$ in $G^\times$ such that the restriction of $f^\times + g^\times$ to $G$ has no excesses on the vertices of $X$.
Finally, we convert $f^\times + g^\times$ to a feasible flow $f^\circ + g^\circ$ in $G^\circ$, from which we get $g^\circ$.

We construct $G^\times$ as follows. Starting with $G^\circ$, we do the following for each vertex $x \in X$:
    \begin{itemize}
    	\item Replace $C_x$ with an arc $(x^{in}, x^{out})$ of capacity $c(x)$.
    	\item Every arc of a capacity $c$ going from a vertex $u$ to a vertex in the cycle $C_x$ is now an arc $(u, x^{in})$ of capacity $c$. 
    	\item Every arc of a capacity $c$ going from a vertex in the cycle $C_x$ to a vertex $x$ is now an arc $(x^{out}, u)$ of capacity $c$. 
    \end{itemize}
In a slight abuse of terminology, we say that a flow in $G^\circ$ is an extension of a flow in $G^\times$ if the two flows have the same restriction to $G$. Similarly, a flow in $G^\times$ is a restriction of a flow in $G^\circ$ if the two flows have the same restriction to $G$.
See Figure~\ref{F:infeasible-vertex}. 
\begin{figure}
\centering
\begin{tabular}{cr@{\qquad}cr@{\qquad}cr}
	\includegraphics[scale=0.45]{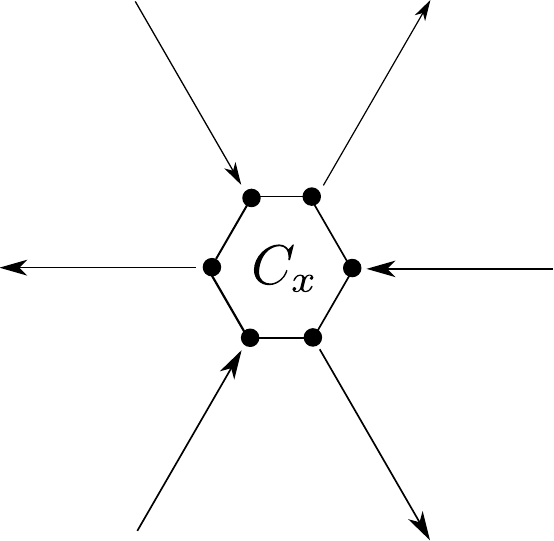} & \hspace{-0.25in} (a)
	&
	\includegraphics[scale=0.45]{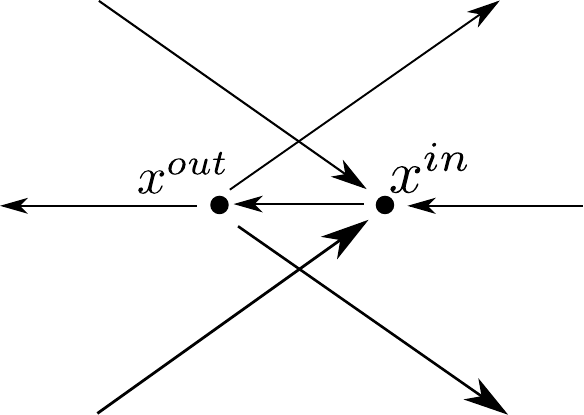} & \hspace{-0.25in}(b)
	&
	\includegraphics[scale=0.45]{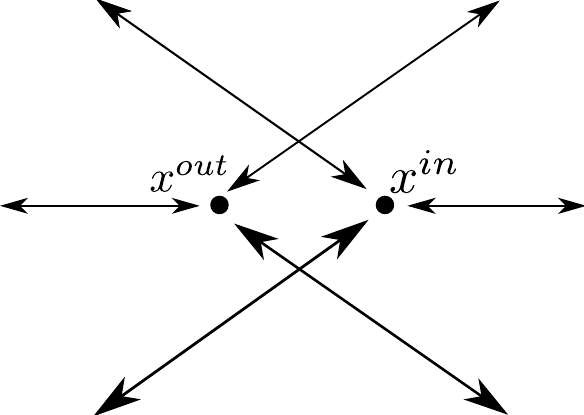} & \hspace{-0.25in} (c)
\end{tabular}
\caption{(a) $C_x$ where $x \in X$ in $G^\circ$ (b) $x^{in}$ and $x^{out}$ in $G^\times$ (c) source $x^{in}$ and sink $x^{out}$ in $H_i$ if $x = x_i$}
\label{F:infeasible-vertex}
\end{figure}
To define $f^\times$, let $f^\times(u,v) = f^\circ(u,v)$ for all arcs $(u,v) \in E(G^\times) \cap E(G^\circ)$, and let $f^\times(x^{in}, x^{out}) = (f^\circ)^{in}(C_x)$ for all $x \in X$.
It is easy to see that $f^\times$ is a flow from $s$ to $t$ whose only infeasible arcs are $(x^{in}, x^{out})$ for all $x \in X$. Furthermore, $(f^\times)^{out}(x^{out}) = (f^\times)^{in}(x^{in}) = f^{in}(x)$ for all $x \in X$. 
We have the following lemma:

\begin{lemma}\label{L:main-integer}
	For each $x \in X$, let $u_x \geq 0$.
    The following two statements are equivalent:
    \begin{enumerate}
    	\item There exists a feasible circulation $g^\circ$ in the residual graph of $G^\circ$ with respect to $f^\circ$ such that
        \[
        	(f^\circ + g^\circ)^{in}(C_{x}) = (f^\circ + g^\circ)^{out}(C_{x}) = f^{in}(x) - u_x
        \]
        for all $x \in X$.
        \item There exists a circulation $g^\times$ in $G^\times$ such that $f^\times + g^\times$ has a feasible extension in $G^\circ$, $f^\times + g^\times$ is feasible in $G^\times$ except possibly at arcs $(x^{in}, x^{out})$ for all $x \in X$, and 
        \[
        	(f^\times + g^\times)^{in}(x^{in}) = (f^\times + g^\times)(x^{in}, x^{out}) = (f^\times + g^\times)^{out}(x^{out}) = f^{in}(x) - u_x
        \]
        for all $x \in X$.
    \end{enumerate}
\end{lemma}
\begin{proof}
    
	$\underline{(1) \Rightarrow (2): }$ Suppose (1) holds.
    Let $g^\times$ be the circulation in $G^\times$ defined by $g^\times(e) = g^\circ(e)$ for each $e \in E(G^\circ) \cap E(G^\times)$ and $g^\times(x^{in}, x^{out}) = (g^\times)^{in}(x^{in})$ for all $x \in X$.
    That is, $g^\times$ is the restriction of $g^\circ$ to $G^\times$.
    The circulation $g^\times$ satisfies conservation constraints at $x^{in}$ by definition, and $g^\times$ satisfies conservation constraints at $x^{out}$ because $g^\times(x^{in}, x^{out}) = (g^\circ)^{in}(C_x) = (g^\circ)^{out}(C_x) = (g^\times)^{out}(x^{out})$.
    Also, $g^\times$ satisfies conservation constraints at all other vertices because $g^\circ$ does.
    
    Since $(f^\circ)^{in}(C_x) = (f^\times)^{in}(x^{in})$ and $(g^\circ)^{in}(C_x) = (g^\times)^{in}(x^{in})$, we have $(f^\circ + g^\circ)^{in}(C_x) = (f^\times + g^\times)^{in}(x^{in})$.    
    A symmetric argument shows that $(f^\circ + g^\circ)^{out}(C_x) = (f^\times + g^\times)^{out}(x^{out})$.
    Flow conservation implies $(f^\times + g^\times)^{in}(x^{in}) = (f^\times + g^\times)(x^{in}, x^{out})$.
    The flow $f^\times + g^\times$ is feasible at all arcs in $E(G^\times) \cap E(G^\circ)$ because $f^\circ + g^\circ$ is.
    
    $\underline{(2) \Rightarrow (1): }$
    Suppose (2) holds.
    There is a feasible extension $h^\circ$ of $f^\times + g^\times$ to $G^\circ$.  
    Let $g^\circ$ be the circulation in $G^\circ$ such that $f^\circ + g^\circ = h^\circ$.
    Since $f^\circ + g^\circ$ is feasible in $G^\circ$, $g^\circ$ is feasible in the residual graph of $G^\circ$ with respect to $f^\circ$.
    It is easy to see that $g^\circ$ is an extension of $g^\times$.
    
    Since $(f^\circ)^{in}(C_x) = (f^\times)^{in}(x^{in})$ and $(g^\circ)^{in}(C_x) = (g^\times)^{in}(x^{in})$, we have $(f^\circ + g^\circ)^{in}(C_x) = (f^\times + g^\times)^{in}(x^{in})$.
    A symmetric argument shows that $(f^\circ + g^\circ)^{out}(C_x) = (f^\times + g^\times)^{out}(x^{out})$.
\end{proof}
If $\lambda \leq \lambda^*$, then there exists a feasible flow $f_\lambda$ in $G$ of value $\lambda$ that can be extended to feasible flows $f_{\lambda}^\times$ in $G^\times$ and $f_{\lambda}^\circ$ in $G^\circ$.
Thus statement (1) of Lemma~\ref{L:main-integer} holds for the circulation $f_{\lambda}^\circ - f^\circ$ in $G^\circ$ and for some choices of $u_x$ where $u_x \geq ex_x$ for all $x \in X$.
Lemma~\ref{L:main-integer} then implies that there exists a circulation $g^\times$ in $G^\times$ such that  $(f^\times + g^\times)^{in}(x^{in}) = (f^\times + g^\times)(x^{in}, x^{out}) = (f^\times + g^\times)^{out}(x^{out}) \leq c(x)$ for all $x \in X$, meaning that $f^\times + g^\times$ is feasible in $G^\times$.
If $\lambda > \lambda^*$, then $g^\times$ may not exist, and the computation of $g^\times$ will fail.
Let $g$ be the restriction of $g^\times$ to $G$.

We will compute the circulation $g^\times$ as the sum of $k-2$ circulations $\phi_1^\times, \dots, \phi_{k-2}^\times$. Let $x_1, \dots, x_{k-2}$ be an arbitrary ordering of the vertices in $X$.
For all $i \in [k-2]$, let $\gamma^\times_i = \phi^\times_1 + \dots + \phi_i^\times$, and let $\gamma_i$ be the restriction of $\gamma_i^\times$ to $G$. 
In particular, $\gamma^\times_0$ is the zero flow and $\gamma^\times_{k-2} = g^\circ$.
We will find the circulations $\phi_1^\times, \dots, \phi_{k-2}^\times$ one by one, and we will maintain the invariant that for all $i \in [k-2]$, $f^\times + \gamma_i^\times$ is a feasible flow in $G^\times$ whose restriction to $G$ has no excess on $x_1, \dots, x_{i}$, at most $\ex(f)$ excess on $x_{i+1}, \dots, x_{k-2}$, and at most $i \cdot \ex(f)$ excess on vertices in $V(G) \setminus X$.
Intuitively, $\phi_i^\times$ gets rid of the excess on $x_i$ without increasing any of the excesses on $x_1, \dots, x_{i-1}$ above 0 and without increasing any of the excesses on $x_{i+1}, \dots, x_{k-2}$ above $\ex(f)$.

So suppose $h^\times$ is a feasible in $G^\times$ whose restriction $h$ to $G$ has no excess on $x_1, \dots, x_{i-1}$, at most $\ex(f)$ excess on $x_i, \dots, x_{k-2}$, and at most $(i-1) \cdot \ex(f)$ excess on all vertices in $V(G) \setminus X$.
Our goal is to find a circulation $\phi_i^\times$ in $G^\times$ such that $h^\times + \phi_i^\times$ is a feasible flow in $G^\times$ whose restriction to $G$ has no excess on $x_1, \dots, x_i$, at most $\ex(f)$ excess on $x_{i+1}, \dots, x_{k-2}$, and at most $i \cdot \ex(f)$ excess on vertices in $V(G) \setminus X$.
For all $i$, finding $\phi_i^\times$ reduces to finding a flow $\phi_{i,H}$ in an $O(k)$-apex graph $H_i$, and we construct $H_i$ as follows: Starting with the residual graph of $G^\times$ with respect to $h^\times$, delete arcs $(x_i^{in}, x_i^{out})$ and $(x_i^{out}, x_i^{in})$. 
Let the source be $x_i^{in}$ and the target be $x_i^{out}$. 
For all $j > i$, if $h^\times(x_j^{in},x_j^{out}) > c(x_j)$, then the arc $(x_j^{in}, x_j^{out})$ has capacity $c(x_j) + \ex(f) - h^\times(x_j^{in}, x_j^{out})$ and the arc $(x_j^{out}, x_j^{in})$ has capacity $h^\times(x_j^{in}, x_j^{out})$.
See Figure~\ref{F:infeasible-vertex}.
We have the following lemma:

\begin{lemma}\label{L:main-integer2}
	Let $u \geq 0$.
    For all $i$, the following two statements are equivalent:
    \begin{enumerate}
        \item There exists a circulation $\phi_i^\times$ in $G^\times$ such that $h^\times + \phi_i^\times$ is feasible in $G^\times$ except possibly at arcs $(x_j^{in}, x_j^{out})$ for all $j > i$, where $(h^\times + \phi_i^\times)(x_j^{in}, x_j^{out}) \leq c(x_j) + \ex(f)$. Also,
        \[
        	(h^\times + \phi_i^\times)^{in}(x_i^{in}) = (h^\times + \phi_i^\times)(x_i^{in}, x_i^{out}) = (h^\times + \phi_i^\times)^{out}(x_i^{out}) = h^{in}(x_i) - u.
        \]
        \item There exists a feasible flow $\phi_{i,H}$ in $H_i$ of value $u$.
    \end{enumerate}
\end{lemma}
\begin{proof}
    $\underline{(1) \Rightarrow (2): }$ 
    Suppose (1) holds. 
    Let $\phi_{i,H}$ be the restriction of $\phi_i^\times$ to $H_i$.
    The flow $\phi_{i,H}$ is feasible in $H_i$ by the definition of $H_i$.
    
    Since $(h^\times)(x_i^{in}, x_i^{out}) = (h)^{in}(x_i)$ and $(h^\times + \phi_i^\times)(x_i^{in}, x_i^{out}) = h^{in}(x_i) - u$, we have $\phi_i^\times(x_i^{out}, x_i^{in}) = u$.
    Since $x_i^{in}$ is not a source in $G^\times$, flow conservation at $x_i^{in}$ implies $(\phi_i^\times)^{out}(x^{in}) = u$.
    This means that $\phi_{i,H}^{out}(x_i^{in}) = u$.
    A symmetric argument implies $\phi_{i,H}^{in}(x_i^{out}) = u$.
    
    $\underline{(2) \Rightarrow (1): }$ 
    Suppose (2) holds.
    Define an extension $\phi_i^\times$ of $\phi_{i,H}$ to a circulation in $G^\times$ by setting $\phi_i^\times(x_i^{out}, x_i^{in}) = u$.
    It is easy to see that $g^\times$ satisfies conservation constraints. The arc capacities in $H_i$ ensure $h^\times + \phi_i^\times$ is feasible in $G^\times$ except possibly at arcs $(x_j^{in}, x_j^{out})$ for all $j > i$, where $(h^\times + \phi_i^\times)(x_j^{in}, x_j^{out}) \leq c(x_j) + \ex(f)$.
    
    Since $x_i^{in}$ is not a terminal in $G^\times$, $h^\times(x_i^{in}, x_i^{out}) = h^{in}(x_i)$, and $\phi_i^\times(x_i^{out}, x_i^{in}) = u$, we have $(h^\times + \phi_i^\times)(x^{in}, x^{out}) = h^{in}(x_i) - u$.
    On the other hand, flow conservation at $x_i^{in}$ and $x_i^{out}$ implies $(h^\times + \phi_i^\times)^{in}(x_i^{in}) = (h^\times + \phi_i^\times)(x_i^{in}, x_i^{out}) = (h^\times + \phi_i^\times)^{in}(x_i^{out})$.
\end{proof}
By the existence of $g^\times$, we know that there exists a circulation $\phi_i^\times$ such that $h^\times + \phi_i^\times$ is feasible in $G^\times$, so statement (1) in Lemma~\ref{L:main-integer2} holds for some $u \geq \ex(h, x_i)$. 
By Lemma~\ref{L:main-integer2}, there must exist a flow $\phi_{i,H}$ of value $\ex(h, x_i)$ in $H_i$.
We compute $\phi_{i,H}$ as follows: Starting with $H_i$, we add a vertex $x^s$ that will be the source instead of $x_i^{in}$, and we add an arc $(x^s, x_i^{in})$ with capacity $\ex(h, x_i)$; similarly, we add a vertex $x^t$ that will be the target instead of $x_i^{in}$, and an arc $(x_i^{out}, x^{t})$ with capacity $\ex(h, x_i)$.
The resulting graph has an acyclic maximum flow that saturates every arc incident to a terminal and so has value $\ex(h, x_i)$, and the restriction of this flow to $H_i$ is $\phi_{i,H}$.
By induction we may assume that $\ex(h^\times, x_i) \leq \ex(f)$, so $v(\phi_{i,H}) \leq \ex(f)$.

By Lemma~\ref{L:main-integer2}, the flow $\phi_{i,H}$ corresponds to a circulation $\phi_i^\times$ in $G^\times$ such that $h^\times + \phi_i^\times$ has no excess on $x_1, \dots, x_i$ and is feasible in $G^\times$ except possibly at arcs $(x_j^{in}, x_j^{out})$ for all $j > i$, where $(h^\times + \phi_i^\times)(x_j^{in}, x_j^{out}) \leq c(x_j) + \ex(f)$.
The restriction of $h^\times + \phi_i^\times$ to $G$ is thus feasible at $x_1, \dots, x_i$ and has at most $\ex(f)$ excess at $x_{i+1}, \dots, x_{k-2}$.
If $\lambda > \lambda^*$, then $\phi_{i,H}$ may not exist, and when we try to compute it, it will have value strictly less than $\ex(h, x_i)$. 
If this happens, then the restriction of $h^\times + \gamma_i^\times$ to $G$ will have positive excess on $x_i$, breaking the desired invariant.
\begin{lemma}\label{L:stage-1-excess}
	Suppose $i \in [k-2]$. If $u \in V(G) \setminus X$, then $\ex(f + \gamma_i, u) \leq i \cdot ex(f)$.
\end{lemma}
\begin{proof}
The proof is by induction on $i$.
The lemma is true if $i = 0$ by the definition of $f^\circ$.
Since $\ex(f + \gamma_{i-1}, u) \leq (i-1) \cdot ex(f)$, we have
	\begin{align*}
    	\ex(f + \gamma_{i}, u) &\leq \ex(f + \gamma_{i-1}, u) + v(\phi_{i,H})\\
		&\leq (i-1) \cdot \ex(f) + \ex(f)\\
        &\leq i \cdot \ex(f)
    \end{align*}
\end{proof}
We have thus shown that for all $i \in [k-2]$, the invariant is maintained: $f^\times + \gamma_i^\times$ is a feasible flow in $G^\times$ whose restriction to $G$ has no excess on $x_1, \dots, x_{i}$, at most $\ex(f)$ excess on $x_{i+1}, \dots, x_{k-2}$, and at most $i \cdot \ex(f)$ excess on vertices in $V(G) \setminus X$.
When $i = k-2$, we get that $f^\times + \gamma_i^\times = f^\times + g^\times$ is a feasible flow in $G^\times$ where $\ex(f + g, u) \leq (k-2) \ex(f)$ for all $u \in V(G) \setminus X$. 
The flow $f^\times + g^\times$ has no excess on the vertices of $X$, so the proof of Lemma~\ref{L:extendable} implies that $f^\times + g^\times$ has an extension in $G^\circ$.
Lemma~\ref{L:main-integer} then implies that $g^\times$ corresponds to a circulation $g^\circ$ in $G^\circ$ such that $\ex(f^\circ + g^\circ, x) = 0$ for all $x \in X$ and $\ex(f^\circ + g^\circ, u) \leq (k-2) \ex(f)$ for all $u \in V(G) \setminus X$; we can compute $g^\circ$ in $O(n \log^3 n)$ time.
We have proved the following lemma:

\begin{lemma}
	For any vertex $x \in X$, $\ex(f^\circ + g^\circ, x) = 0$. For any vertex $v \in V(G) \setminus X$, $\ex(f^\circ + g^\circ, v) \leq (k-2)\ex(f)$.
\end{lemma}

Computing $\phi_i^\times$ requires us to compute a maximum flow in a graph with $O(k)$ apices (these are $s$, $t$, and $x^{in}$ and $x^{out}$ for all $x \in X$), which takes $O(k^3 n \log^3 n)$ time using the algorithm of Borradaile et al.~\cite{BKMNW17}. Since we need to compute $k-2$ such flows, computing $g^\circ$ takes $O(k^4 n \log^3 n)$ time.

\subsection{Stage 2}
Having found an integer circulation $g^\circ$ in $G^\circ$, we construct the fractional circulation $g^\circ/k$ in $G^\circ$.
Using the algorithm of Lemma~\ref{L:fractional-alg}, we can let $g_k^\circ$ be an integer circulation in $G^\circ$ such that $|(g^\circ/k)(e) - g_k^\circ(e)| < 1$ for every arc $e$ in $G^\circ$.

\begin{lemma}\label{L:reduce-excess}
	For any vertex $v \in V(G)$, $\ex(f^\circ + g_k^\circ, v) \leq \frac{k-1}{k}\ex(f) + \Delta$.
\end{lemma}
\begin{proof}
For any vertex $x \in X$, we have $\ex(f^\circ + g^\circ, x) = 0$, so $\ex(f^\circ + g^\circ/k, x) \leq \frac{k-1}{k}ex_x$.
By Lemma~\ref{L:fractional-alg}, $g_k^\circ(e)$ and $(g^\circ/k)(e)$ differ by at most 1 on every arc $e$.
	There are $\deg_G(x)$ arcs in $G^\circ$ that are incident to at least one vertex in $C_x$, so 
\[
	\ex(f^\circ + g_k^\circ, x) \leq \frac{k-1}{k} ex_x + \deg_G(x) \leq \frac{k-1}{k} \ex(f) + \Delta.
\]

For any vertex $v \in V(G) \setminus X$, we have $\ex(f^\circ, v) = 0$ and $\ex(f^\circ + g^\circ, v) \leq (k-2)ex(f)$ by Lemma~\ref{L:stage-1-excess}.
This implies that $\ex(f^\circ + g^\circ/k, v) \leq \frac{k-2}{k}ex(f)$.
By Lemma~\ref{L:fractional-alg}, $g_k^\circ(e)$ and $(g^\circ/k)(e)$ differ by at most 1 on every arc $e$.
	There are $\deg_G(v)$ arcs in $G^\circ$ that are incident to at least one vertex in $C_v$, so 
\[
	\ex(f^\circ + g_k^\circ, v) \leq \frac{k-2}{k} \ex(f) + \deg_G(v) \leq \frac{k-1}{k} \ex(f) + \Delta.
\]
\end{proof}

Using the algorithm of Lemma~\ref{L:fractional-alg}, computing $g_k^\circ$ takes $O(n \log^3 n)$ time.  

\subsection{Stage 3}
In this stage, we finally get $f_1^\circ$. 
Using Lemma~\ref{L:cancel-cycles}, we find a flow $f_1^\circ$ of the same value as $f^\circ + g_k^\circ$ such that the restriction $f_1$ of $f_1^\circ$ to $G$ is acyclic.
By Lemma~\ref{L:few-saddles}, $f_1^\circ$ has at most $k-2$ infeasible vertices.
Since $f_1(e) \leq (f^\circ + g_k^\circ)(e)$ for all arcs $e$, we still have $\ex(f^\circ + g_k^\circ, v) \leq \frac{k-1}{k}\ex(f) + \Delta$ for all vertices $v \in V(G)$.
Stage 3 takes $O(n)$ time, so the total running time of stages 1-3 is $O(k^4 n \log^3 n)$.

\section{The case $k = 3$}
In the case of three terminals, we can find a maximum flow in $O(n \log n)$ time even if $G$ has arbitrary real capacities.
Without loss of generality, we may assume that there are two sources and one sink.
Let $f^\circ$ be a maximum flow in $G^\circ$.
We can compute $f^\circ$ in $O(n \log n)$ time by using the algorithm of Borradaile and Klein~\cite{BK09}~\cite{E10}: first find a maximum flow $f^\circ_1$ from $s_1$ to $t$, and then find a maximum flow $f_2^\circ$ from $s_2$ to $t$ in the residual graph of $G^\circ$ with respect to $f^\circ_1$.
The desired flow $f^\circ$ is just $f^\circ_1 + f^\circ_2$.
By Lemma~\ref{L:cancel-cycles} we may assume without loss of generality that the restriction $f$ of $f^\circ$ to $G$ is acyclic.
By Lemma~\ref{L:few-saddles}, the flow graph $f_G$ of $f$ has at most one saddle $x$, and has index 1.
If $f$ is feasible at $x$, then $f$ is the maximum flow in $G$, so assume $f$ is infeasible at $x$.

\subsection{Almost-feasible flows}
Let $\delta = val(G^\circ) - val(G)$.
Suppose $f_\delta^\circ$ is a maximum flow in $G^\circ$ whose restriction $f_\delta$ to $G$ is acyclic and has a single infeasible vertex $x_\delta$.
If $\ex(f_\delta,x_\delta) = \delta$, then $f_\delta^\circ$ and $f_\delta$ are {\em almost feasible}.
Given an almost-feasible flow $f_\delta$ in $G$, we can remove $\delta$ units of flow through $x_\delta$ to get a maximum flow in $G$. This can be done in $O(n \log n)$ time using the algorithm of Borradaile and Klein~\cite{BK09} for finding maximum flow in planar graphs, as follows. 
Treat the arc weights in the flow graph $f_{\delta,G}$ of $f_\delta$ as capacities.
In $f_{\delta,G}$, find a flow $g_1$ of value $\min\{f_\delta^{out}(s_1), \delta\}$ from $s_1$ to $x_\delta$, and find a flow $g_2$ of value $\delta - v(g_1)$ from $s_2$ to $x_\delta$ in the residual graph of $f_{\delta,G}$ with respect to $g_1$; the result is that $g_1 + g_2$ is a flow from $s_1$ and $s_2$ to $x_\delta$ of value $\delta$.
Next, in $f_{\delta,G}$, we find a flow $g_3$ of value $\delta$ from $x_\delta$ to $t$. (As usual, to find a flow of a certain value from a source $s$ to a sink $t$, we add a vertex $s'$ that will be the source instead of $s$, and we add an arc $(s', s)$ whose capacity is the desired flow value.)
Since $f_\delta$ is acyclic, the supports of $g_1 + g_2$ and $g_3$ do not share any arcs. 
Then, for any arc $e \in E(G)$, reduce $f_\delta(e)$ by $g_1(e) + g_2(e) + g_3(e)$. 
We have removed $\delta$ units of flow through $x_\delta$ in the flow $f_\delta$.

In this subsection, we show that almost-feasible flows exist.

\begin{theorem} There exists a maximum flow $f^\circ_\delta$ in $G^\circ$ such that the restriction $f_\delta$ of $f^\circ_\delta$ to $G$ is acyclic and has a single infeasible vertex $x$ with $\ex(f_\delta,x) = \delta$. 
\end{theorem}
\begin{proof}
Let $f_{max}$ be a maximum flow in $G$, and let $f_{max}^\circ$ be an extension of $f_{max}$ to $G^\circ$.
In the residual graph of $G^\circ$ with respect to $f_{max}^\circ$, find an acyclic maximum flow $g$. Let $(f')^\circ = f^\circ + g$ and let $f'$ be the restriction of $(f')^\circ$ to $G$.
Since $g$ has value $\delta$, the excess of every vertex under $f'$ is at most $\delta$.
Using the algorithm of Lemma~\ref{L:cancel-cycles}, compute a flow $f^\circ_\delta$ with the same value as $(f')^\circ$ whose restriction $f_\delta$ to $G$ does not contain flow-cycles.

Since $g^\circ$ is a maximum flow in the residual graph of $G^\circ$ with respect to $f^\circ$, $(f')^\circ$ and $f_\delta^\circ$ are maximum flows in $G^\circ$. 
Since $f_\delta(e) \leq f'(e)$ for every arc $e$, we have $\ex(f_\delta) \leq \delta$.
Since $f_\delta$ is acyclic, Lemma~\ref{L:cancel-cycles} implies that $f_\delta$ has at most one infeasible vertex $x$.

If $\ex(f_\delta, x) < \delta$, then, starting with $f_\delta$, we can remove $\ex(f_\delta, x)$ units of flow through $x$ to get a feasible flow in $G$ with value strictly higher than $v(f_\delta) - \delta = val(G)$, a contradiction.
Thus $\ex(f_\delta,x) = \delta$.
\end{proof}



\subsection{Getting an almost-feasible flow}
It remains to show how to compute an almost-feasible flow.
We will describe an algorithm that finds a circulation $g^\circ$ in $G^\circ$ such that the restriction of $f^\circ + g^\circ$ to $G$ is almost feasible. Let $ex_x = \ex(f)$.

%
%
%
%

We construct a flow network $H$ from $G^\circ$ and $f^\circ$ in the exact same way as in section~\ref{S:integer-h}.
That is, starting with the residual graph of $G^\circ$ with respect to $f^\circ$, we do the following:
\begin{itemize}
	\item
	Replace $C_x$ with vertices $x^{in}$ and $x^{out}$.
    \item 
    Every arc of capacity $c$ going from a vertex $u$ to a vertex in the cycle $C_x$ is now an arc from $(u, x^{in})$ of capacity $c$
    \item
    Every arc of capacity $c$ going from a vertex in the cycle $C_x$ to a vertex $u$ is now an arc $(x^{out},u)$ of capacity $c$.
	\item
	Let $x^{in}$ be the source and $x^{out}$ be the sink.
\end{itemize}
Lemmas~\ref{L:main-integer} and~\ref{L:main-integer2} both apply.
Thus our goal is now to find a maximum flow $g_H$ in $H$ that can be extended to a circulation in the residual graph of $G^\circ$ with respect to $f^\circ$.

    
    
    
    
We will now show that we can make two simplifications to $H$.
The goal of these simplifications is to eliminate the apices $s, x^{in}$, and $x^{out}$ so that $H$ becomes planar.
First, since our goal is a flow in $H$ from $x^{in}$ to $x^{out}$, we may assume without loss of generality that arcs of the form $(u, x^{in})$ and $(x^{out}, v)$ do not exist in $H$. 
As a result, the only arcs in $H$ that are incident to $x^{in}$ are arcs of the form $(x^{in}, u)$ where $f(u, x) > 0$.
If we consider these arcs as arcs in $G$, then, since $x$ has index 1, these arcs form two intervals in the cyclic order around $x$.
Therefore, we can replace $x^{in}$ with two sources $x^{in}_1$ and $x^{in}_2$, replacing arcs $(x^{in}, u)$ in the first interval with $(x^{in}_1,u)$ and arcs $(x^{in}, u)$ in the second interval with arcs $x^{in}_2$. A similar simplification eliminates $x^{out}$. See Figure~\ref{F:split-terminals}. (We could not perform this simplification in Section~\ref{S:integer-h} because the desired flow in $H$ could send flow from $x^{in}$ to $(x')^{in}$ to $x^{out}$.)
One effect of this simplification is that every flow $g_H$ in $H$ automatically extends to a circulation $g^\circ$ in the residual graph of $G^\circ$ with respect to $f^\circ$.
This is because $f$ has an extension in $G^\circ$ and $(f + g_H)(e) \leq f(e)$ for any arc $e$ incident to $x$ in $G$, so $f + g_H$ has an extension to $G^\circ$.

\begin{figure}
\centering
\begin{tabular}{cr@{\qquad}cr}
	\includegraphics[scale=0.5]{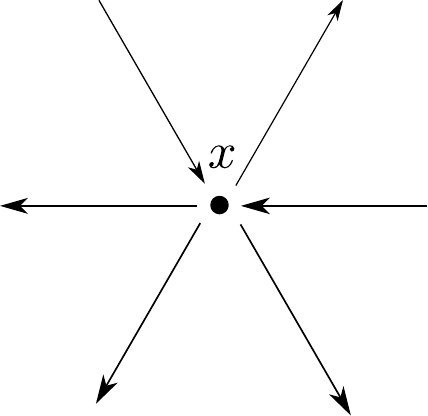} & \hspace{-0.25in} (a)
	&
	\includegraphics[scale=0.5]{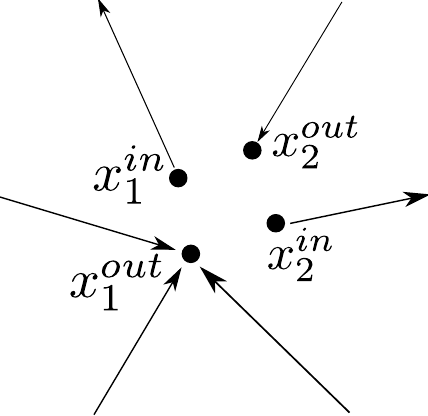} & \hspace{-0.25in}(b)
\end{tabular}
\caption{(a) The flow graph of $f$ at the unique infeasible vertex (b) sources and sinks of $H$}
\label{F:split-terminals}
\end{figure}

Second, we show that we can delete the arcs $(s,s_1)$ and $(s,s_2)$. This eliminates the apex $s$.
\begin{lemma}\label{L:avoid-s}
    If there is an augmenting path in $H$ (i.e., a path from a source to a target in $H$) containing $s$, then there is an augmenting path in $H$ not containing $s$.
\end{lemma}
\begin{proof}
	See Figure~\ref{F:avoid-s}. Consider two arcs $e$ and $e'$ carrying flow out of $x$ such that as we cyclically traverse the arcs incident to $x$ in clockwise order, some arc between $e$ and $e'$ carries flow into $x$, and some arc between $e'$ and $e$ carries flow into $x$.
    There must be a path $P$ from $x$ to $t$ starting with $e$ that carries flow.
    Similarly, there must be a path $P'$ from $x$ to $t$ starting with $e'$ that carries flow.
    Without loss of generality, assume $P$ and $P'$ do not cross. Let $u$ be the first vertex on $P$ after $x$ that also appears on $P'$.
    The vertex $u$ must also be the first vertex on $P'$ after $x$ that also appears on $P$, because otherwise $f$ has flow-cycles. 
    Let $Q$ be the prefix of $P$ that ends at the arc of $P$ that goes into $u$, and let $Q'$ be the prefix of $P'$ that ends at the arc of $P'$ that goes into $u$. 
    These prefixes are well-defined because $f$ is acyclic.
    Since both $Q$ and $Q'$ go from $x$ to $u$, their union partitions the plane into two regions.
    Denote the inner region by $R$ and the outer region by $R'$.
    
    Since there are arcs in both $R$ and $R'$ carrying flow into $x$ and $f$ is acyclic, one source must be in $R$ and the other must be in $R'$.
    Furthermore, there is some path $Q_s$ from $s_2$ to $x$ carrying flow, and there is some path from $s_1$ to $x$ carrying flow.
   
    Without loss of generality, suppose the augmenting path $\pi$ in $H$ starts in $x^{in}$, goes to $s_1 \in R$, uses arcs $(s_1, s)$ and $(s, s_2)$, and ends by going from $s_2 \in R'$ to $x^{out}$.
    We can replace it by an augmenting path $\pi'$ that starts at $x^{in}$, follows $rev(Q_s)$ to $s_2$, and then follows $\pi$ from $s_2$ to $x^{out}$.
    The augmenting path $\pi'$ does not contain $s$.
    
\begin{figure}
\centering
\includegraphics[scale=0.5]{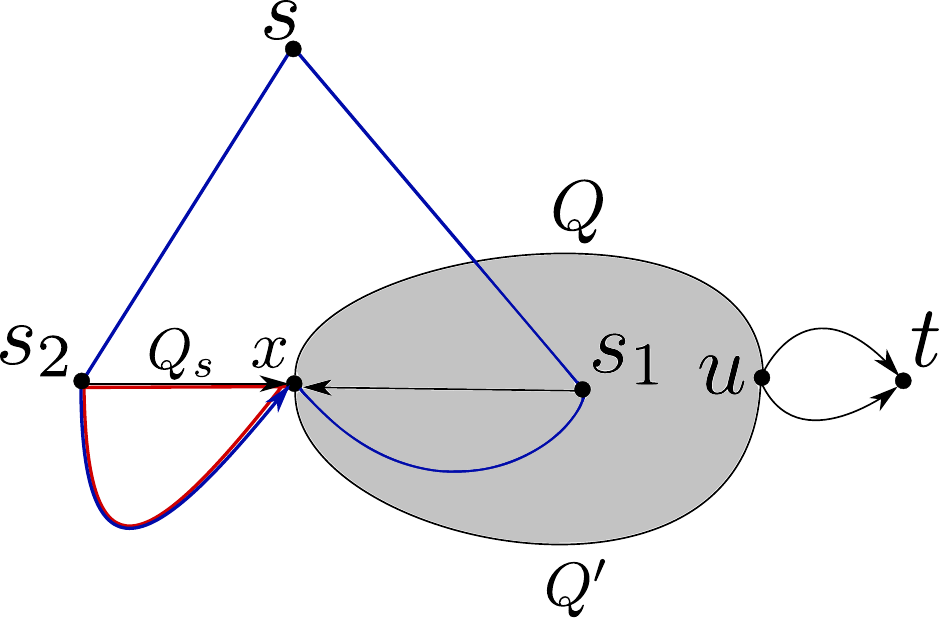}
\caption{$H$ in the proof of Lemma~\ref{L:avoid-s} with all terminals merged into the vertex $x$. The blue path is $\pi$. The red path is $\pi'$. The shaded region is $R$.}
\label{F:avoid-s}
\end{figure}
\end{proof}

Let $g_H$ be the maximum flow in $H$ and let $g^\circ$ be its extension to $G^\circ$.
We apply Lemma~\ref{L:cancel-cycles} to find a flow $f^\circ_3$ with the same value as $f^\circ + g^\circ$ whose restriction $f_3$ has no flow-cycles.

By Lemma~\ref{L:few-saddles}, the flow $f^\circ_3$ is infeasible at a single vertex $y$.
If $y=x$, then $f^\circ_3$ must be almost feasible. 
This is because Lemma~\ref{L:main-integer} implies that if $g_H$ is a maximum flow in $H$, then $f^\circ + g^\circ$ is a maximum flow in $G^\circ$ that minimizes the excess of $x$.
Furthermore, $f_3(e) \leq (f^\circ + g^\circ)(e)$, so $f_3$ is a maximum flow in $G^\circ$ that minimizes the excess of $x$.

If $y \neq x$, then we define a function $F : E(G) \times [0, 1] \to \mathbb{R}$ for each arc $e \in G$. $F(e, \beta)$ is defined as follows. We apply Lemma~\ref{L:cancel-cycles} to $f^\circ + \beta g^\circ$ to get a flow $f_\beta^\circ$ whose restriction $f_\beta$ to $G$ is acyclic.
We then define $F(e, \beta) = f_\beta(e)$.
For all arcs $e \in E(G)$, we have $F(e, 0) = f(e)$ and $F(e, 1) = f_3(e)$.

Clearly, $F(\cdot, \beta)$ has an extension that is feasible in $G^\circ$ for all $\beta$, and $F(e, \cdot)$ is continuous for any arc $e \in E(G)$.
Consider how $F(\cdot, \beta)$ changes as $\beta$ increases from 0 to $1$.
We start with excess on $x$ and no other vertices, and end with excess on $y$ but no other vertices.
Moreover, no matter what $\beta$ is, there is at most one infeasible vertex.
Thus, at some point, say when $\beta = \beta_0$, we must have no infeasible vertices.
Since $F(\cdot, \alpha_0)$ is a maximum flow in $G^\circ$, it must be a maximum flow in $G$.

To compute $\beta_0$, we need the following lemma.
\begin{theorem}\label{L:constant-derivative}
For every fixed arc $e \in E(G)$, $\frac{\partial F(e, \beta)}{\partial \beta}$ is constant.
\end{theorem}
\begin{proof}
The proof requires understanding the details of the algorithm of Lemma~\ref{L:cancel-cycles}, which can be found in Appendix~\ref{A:cancel-cycles}.
Here we summarize how the flow $F(\cdot, \beta)$ is computed: 
\begin{enumerate}
	\item Compute $f^\circ_\beta = f^\circ + \beta g^\circ$.
    Define a capacity function $c'$ by $c'(e) = f^\circ_\beta(e)$ for all $e \in E(G)$ and $c'(e) = c(e)$ for all $e \notin E(G)$.
    Construct the residual graph $G^\circ_\beta$ of $G^\circ$ with respect to $f^\circ_\beta$ and $c'$.
    Let $h_\infty$ be the infinite face of $G^\circ \setminus \{s\}$.
	For each face $h$ of $G^\circ_\beta \setminus \{s\}$, let $\Phi(h)$ be the distance of $h^*$ from $h_\infty^*$ in $(G^\circ_\beta \setminus \{s\})^*$. For each arc $e$ in $G^\circ \setminus \{s\}$, let $h_\ell$ be the face on the left of $e$ and let $h_\ell$ be the face on the right. Let $g_\beta(e) = \Phi(h_r) - \Phi(h_\ell)$ for each arc $e$ in $G^\circ \setminus \{s\}$; $g_\beta$ is a simple circulation.
	Finally, let $f_\gamma^\circ = f^\circ_\beta + g_\beta$, and let $f_\gamma$ be the restriction of $f_\gamma^\circ$ to $G$.
	The flow $f_\gamma$ has no counter-clockwise flow-cycles. 	
    \item 
    Define a new capacity function $c''(e) = f^\circ_\gamma(e)$ for $e \in E(G)$ and $c''(e) = c(e)$ for $e \notin E(G)$.
	Construct the residual graph $G^\circ_\gamma$ of $G^\circ$ with respect to $c''$ and $f^\circ_\gamma$.
	For each face $h$ of $G^\circ_\gamma \setminus \{s\}$, let $\Phi(h)$ be he distance of $h^*$ from $h_\infty^*$ in $(G^\circ_\gamma \setminus \{s\})^*$.
    Let $g_\gamma(e) = \Phi(h_\ell) - \Phi(h_r)$. 
	Finally, $F(\cdot, \beta)$ is the restriction of $f^\circ_\gamma + g_\gamma$ to $G$.
\end{enumerate}

It suffices to show that the shortest path trees $T_\beta$ in $(G^\circ_\beta \setminus \{s\})^*$ and $T_\gamma$ in $(G^\circ_\gamma \setminus \{s\})^*$ rooted at $h_\infty^*$ do not change as $\beta$ increases. Suppose for the sake of argument that $T_\beta$ changes as $\beta$ increases.
Then, there exist vertices $u^*$ and $v^*$ in $(G^\circ_\beta \setminus \{s\})^*$ and two internally disjoint paths $P_1^*$ and $P_2^*$ from $u^*$ to $v^*$ in $(G^\circ_\alpha \setminus \{s\})^*$ whose lengths are changing at different rates as $\beta$ increases.
Let $H$ be the region bounded by $P_1^*$ and $P_2^*$, and suppose that $P_1 \circ rev(P_2)$ is a clockwise cycle.
The change in the length of $P_1^*$ in $(G^\circ_\beta \setminus \{s\})^*$ is the change in the capacity of the cut $P_1$ in $G^\circ_\beta \setminus \{s\}$, which is the change in the amount of flow $f^\circ + \beta g^\circ$ sends out of $H$ through the arcs of $P_1$. 
Similarly, the change in the length of $P_2^*$ is the change in the amount of flow $f^\circ + \beta g^\circ$ sends into $H$ through the arcs of $P_2$.
This means that the net amount of flow that $f^\circ + \beta g^\circ$ carries into $H$ is changing as $\beta$ increases, but this is impossible, since $g^\circ$ is a simple circulation.
We conclude that $T_\beta$ does not increase as $\beta$ increases.
A similar argument shows that since $g_\beta$ is a simple circulation and $f_\gamma^\circ = f^\circ + \beta g^\circ + g_\beta$, $T_\gamma$ does not change as $\beta$ increases. 


\end{proof}

The previous lemma implies that $\frac{d}{d\beta} ex(F(\cdot, \beta),x)$ 
is constant, and we can find it because
\[
	\frac{d}{d\beta} ex(F(\cdot, \beta),x)  = \ex(F(\cdot, 1),x) - \ex(F(\cdot, 0),x) = \ex(f_3,x) - \ex(f,x),
\]
We then let
\[
	\beta_0 = -\frac{\ex(F(\cdot, 0),x)}{\frac{d}{d\beta} \ex(F(\cdot, \beta),x)}.
\]
and $F(\cdot, \beta_0)$ is a maximum flow in $G$.

The algorithm takes $O(n \log n)$ time to compute $f^\circ$.
It takes $O(n \log n)$ time to compute $g_H$, from which we can obtain $g^\circ$, $f_3^\circ$, and $f_3$ in linear time.
If $y = x$, then we have an almost-feasible flow that can be turned into a maximum flow in $G$ in $O(n \log n)$ time. If $y \neq x$, then we can compute $\beta_0$ and $F(\cdot, \beta_0)$ in linear time.
The entire algorithm takes $O(n \log n)$ time.
\subsection{Discussion}
One natural question is what happens when $k = 4$.
Here, we can define a maximum flow in $G^\circ$ as being almost feasible if we can remove $\delta$ units of flow to get a feasible flow in $G$.
We can also prove that almost-feasible flows always exist.
The main problem seems to be that there is no easy way of characterizing or getting almost-feasible flow. 
For example, minimizing the sum of the excesses of the two infeasible vertices $x$ and $x'$ does not necessarily work.
Suppose there is one flow where the infeasible vertices both have excesses of 10, and another flow where the excesses are both 7. 
If $\delta = 10$, then it could be the case that the first flow is almost feasible because removing a unit of flow through $x$ may simultaneously remove a unit of flow through $x'$ (i.e., we can decompose the first flow into paths and cycles such that some paths pass through both $x$ and $x'$), while the second flow is not almost feasible because removing a unit of flow through $x$ does not simultaneously remove a unit of flow through $x'$, and vice versa.
\bigskip
\noindent {\bf Acknowledgments.} I would like to thank Jeff Erickson for helpful discussions and for comments on an earlier draft of this paper.

\bibliography{references}{}
\bibliographystyle{plain}

\begin{appendices}
\section{Proof of Lemma~\ref{L:fractional}}\label{A:fractional}
For every arc $e$ in $G^\circ$, let $\delta(e)$ be the fractional part of $f^\circ(e)$.
We define the {\em fractional residual graph} of $G^\circ$ with respect to $f^\circ$ as follows: starting with the vertices of $G^\circ$, for every arc $e$ with $\delta(e) > 0$, we add an arc $e$ with capacity $1 - \delta(e)$ and edge $rev(e)$ with capacity $\delta(e)$; these are the only arcs in the fractional residual graph.
Let $H$ be the fractional residual graph of $G^\circ$ with respect to $f^\circ$.
The following algorithm finds $f_1$:
\begin{itemize}
	\item While $H$ contains at least one arc: 
    \begin{itemize}
    	\item Find a feasible flow-cycle $g^\circ$ in $H$ such that $g^\circ$ saturates some arc $e$ in $H$.
        \item Replace $f^\circ$ with $f^\circ + g^\circ$, and update $H$ accordingly.
	\end{itemize}
    \item Set $f_1 = f$.
\end{itemize} 
We need to show that if $H$ contains at least one arc, then it contains a cycle.
Suppose, for the sake of argument, that $H$ is a tree.
Let $P = (u_1, v_1), \dots, (u_{p}, v_p)$ be a maximal path in $H$, meaning that there is no arc entering $u_1$ from a vertex outside $P$. 
This means that in $f^\circ$, $(u_1, v_1)$ is the only arc incident to $u_1$ carrying fractional flow.
If $u_1$ is one of the terminals $s$ or $t$, then this violates the fact that $v(f^\circ)$ is an integer. 
If $u_1$ is not a terminal, then this violates the fact that flow is conserved at $u_1$.
In both cases, we get a contradiction, so $H$ has a cycle.

Every time we update $f^\circ$, at least two arcs, namely the arcs $e$ and $rev(e)$, disappear from the fractional residual graph, and we never add any arcs to the fractional residual graph, so the algorithm eventually terminates.
When the fractional residual graph has no arcs, $f$ is an integer flow.
Also, for every arc $e$ in $G^\circ$, $f^\circ(e)$ never decreases below $\lfloor f^\circ(e) \rfloor$ or increases above $\lceil f^\circ(e) \rceil$, so $|f^\circ(e) - f_1^\circ(e)| < 1$ for all arcs $e$.

\section{Proof sketch of Lemma~\ref{L:cancel-cycles}}\label{A:cancel-cycles}
The purpose of this section is to describe the algorithm of Lemma~\ref{L:cancel-cycles}. This is needed for the proof of Lemma~\ref{L:constant-derivative}. We will not prove the correctness of the algorithm, as that has been done elsewhere~\cite{KNK93}~\cite{KN11}.
\subsection{Duality}
First, we need a few standard definitions.
If $G$ is a planar graph, the {\em dual graph} $G^*$ of $G$ has a vertex $h^*$ for every face $h$ of $G$, and an arc $e^*$ for every arc $e$ of $G$.
The arc $e^*$ is directed from the vertex of $G^*$ corresponding to the face in $G$ on the left side of $e$, to the vertex of $G^*$ corresponding to the face in $G$ on the right side of $e$.
If $e$ is undirected, then so is $e^*$.
Any undirected edge $\{u,v\}$ can be represented by two directed arcs $(u,v)$ and $(v,u)$, each with the same weight as $\{u,v\}$.
We put lengths $\ell(e^*)$ on the edges $e^*$ of $G^*$ as follows: $\ell(e^*) = c(e)$ for every $e \in E(G)$.

\subsection{Algorithm description}
The algorithm has three steps and is based on an algorithm of Khuller, Naor, and Klein~\cite{KNK93} that finds a circulation without clockwise residual cycles in a directed planar graph in $O(n)$ time.

{\bf Finding a circulation without clockwise residual cycles.}
We describe the algorithm of Khuller, Naor, and Klein that finds a circulation $g$ in $G^\circ$ without clockwise residual cycles~\cite{KNK93}.

The graph $G^\circ \setminus \{s,t\}$ is planar. Let $h_\infty$ be the infinite face of $G^\circ \setminus \{s,t\}$, and let $h_\infty^*$ be its dual vertex.
Using the algorithm of Henzinger et al.~\cite{HKRS97}, compute the shortest path tree rooted at $h^*_\infty$ in $(G^\circ \setminus \{s,t\})^*$ in $O(n)$ time.
For every face $h$ of $G$, let $\Phi(h)$ be the distance in $(G^\circ \setminus \{s,t\})^*$ from $h_\infty^*$ to $h^*$.
For any edge $e \in E(G^\circ \setminus \{s,t\})$, we define $g(e)$ as follows.
Let $h_\ell$ be the face on the left of $e$ and $h_r$ be the face on the right of $e$.
If $\Phi(h_r) \geq \Phi(h_\ell)$, then set $g(e) = \Phi(h_r) - \Phi(h_\ell)$.
Otherwise, set $g(e) = 0$ (and $g(rev(e))$ will automatically be set to $\Phi(h_\ell) - \Phi(h_r))$.
Khuller, Naor, and Klein~\cite{KNK93} proved that the resulting flow function $g$ is a simple circulation in $G^\circ$ such that $G^\circ$ has no clockwise residual cycles with respect to $g$.




{\bf Finding a flow without clockwise residual cycles.}
Let $f^\circ$ be a feasible flow in $G^\circ$.
We describe an algorithm due to Kaplan and Nussbaum~\cite{KN11} that computes a flow $f_1^\circ$ in $G^\circ$ with the same value as $f^\circ$ and without clockwise residual cycles.
A symmetric algorithm can then compute a flow in $G^\circ$ with the same value as $f^\circ$ and without counterclockwise residual cycles.

Let $G^\circ_{f}$ be the residual graph of $G^\circ$ with respect to $f^\circ$.
Using the algorithm of step 1, find a circulation $g$ in $G^\circ_{f}$ such that $G^\circ_{f}$ does not have clockwise residual cycles with respect to $g$.
Now define $f^\circ_1 = f^\circ + g$.
Computing $f_1^\circ$ takes $O(n)$ time.
Kaplan and Nussbaum showed that $f_1^\circ$ is a feasible flow in $G^\circ$ with the same value as $f^\circ$ and without clockwise residual cycles~\cite{KN11}.


{\bf Finding an acyclic flow.}
Finally, let $f^\circ$ be a feasible flow in $G^\circ$.
We describe the algorithm due to Kaplan and Nussbaum~\cite{KN11} that computes a flow of the same value as $f^\circ$ whose restriction to $G$ is acyclic.
We will do this by first eliminating counterclockwise flow-cycles to get a flow $f^\circ_1$; a symmetric algorithm then eliminates clockwise flow-cycles.
 
Define a new capacity function $c_1$ on the arcs of $G^\circ$ by first setting $c_1(e) = f^\circ(e)$ for $e \in E(G)$.
This will ensure that we do not increase the flow along any arc of $G$. 
All other arcs in $G^\circ$ are in $C_v$ for some vertex $v$; for these arcs $e$ we set $c_1(e) = c(e) = c(v)/2$.
Now we apply the previous algorithm to $G^\circ$ and $c_1$ to find a flow $f^\circ_1$ with the same value as $f^\circ$ such that there are no clockwise residual cycles in $G^\circ$ with respect to $f^\circ_1$ and $c_1$.
Kaplan and Nussbaum~\cite{KN11} showed that	the restriction of $f^\circ_1$ to $G$ does not contain counterclockwise flow-cycles.

We now repeat the previous procedure symmetrically, by defining a new capacity $c_2$ that restricts the flow on every arc $e$ of $G$ to be at most $f^\circ_1(e)$, and finding a circulation in $G^\circ$ without counterclockwise residual cycles.
This way we get from $f^\circ_1$ a flow $f^\circ_2$ of the same value whose restriction to $G$ does not contain clockwise flow-cycles in $G$.
For every $e \in E(G)$, we have $f^\circ_2(e) \leq f^\circ_1(e) \leq f^\circ(e)$, so we did not create any new flow-cycles when going from $f^\circ$ to $f^\circ_1$ to $f^\circ_2$.
Thus $f^\circ_2$ is a feasible flow in $G^\circ$ with the same value as $f^\circ$ whose restriction to $G$ is feasible and acyclic.

\section{Proof of Lemma~\ref{L:few-saddles}}\label{A:saddles-and-excess}
In this section, we prove Lemma~\ref{L:few-saddles}.

First we need a few definitions. For any face $\phi$ in $f_G$, let $\alpha(\phi)$ denote the alternation number of $\phi$; $\alpha(\phi)$ is the number of times the arcs on the boundary of $\phi$ change direction as we traverse this boundary.
Thus $\alpha(\phi) = 0$ if the arcs on the boundary of $\phi$ form a directed cycle. 
We use $index(\phi)$ to denote the index of a face $\phi$, which is defined by $index(\phi) = \alpha(\phi)/2 - 1$.

Now we can proceed with the proof. 
See Figure~\ref{F:pie}. 
If at each vertex $v$ in $f_G$ we cycle through its incident arcs in order according to the embedding of $f_G$, each transition from one arc $e$ to the next arc $e'$ results in exactly one alternation either for $v$ or for the face on whose boundary the two arcs $e$ and $e'$ lie. Thus
    \begin{align*}
    	2E &= \sum_{v \in V(f_G)} \alpha(v) + \sum_{\phi \in F(f_G)} \alpha(\phi)\\
        \implies E &= \sum_{v \in V(f_G)} (index(v) + 1) + \sum_{\phi \in F(f_G)} (index(\phi) + 1)\\
        \implies E &= \sum_{v \in V(f_G)} index(v) + \sum_{\phi \in F(f_G)} index(\phi) + V + F\\
        \implies -2 &= \sum_{v \in V(f_G)} index(v) + \sum_{\phi \in F(f_G)} index(\phi)
    \end{align*}
    where in the last line we have used Euler's formula $V(f_G) - E(f_G) + F(f_G) = 2$.
    Since $f_G$ is acyclic, $index(\phi) \geq 0$ for each face $\phi$, so $-2 \geq \sum_{v \in V(f_G)} index(v)$.
    Finally, $index(v) = -1$ for each terminal $v$, so
    \[
    	k_1 + k_2 - 2 \geq \sum_{v: index(v) \geq 1} index(v).
    \]
    A vertex $v$ is a saddle if and only if $index(v) \geq 1$, so this shows that the sum of the indices of the saddles in $f_G$ is at most $k_1 + k_2 - 2$.
    \begin{figure}
\centering
\includegraphics[scale=0.5]{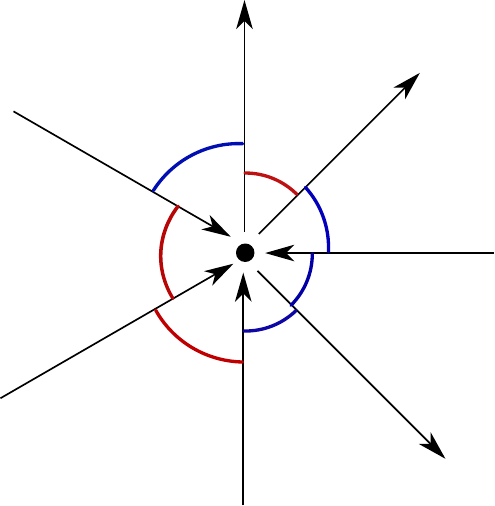}
\caption{Proof of Lemma~\ref{L:few-saddles}. Blue transitions contribute one alternation to a vertex; red transitions contribute one alternation to a face.}
\label{F:pie}
\end{figure}

\end{appendices}
\end{document}